\documentclass[11pt,a4paper]{article}
\usepackage{authblk}
\usepackage[T1]{fontenc}
\usepackage{amsmath,amssymb,amsthm}
\usepackage{geometry}
\usepackage{enumitem}
\usepackage{xcolor}
\usepackage{hyperref}
\usepackage{authblk}

\numberwithin{equation}{section}

\theoremstyle{plain}
\newtheorem{theorem}{Theorem}[section]
\newtheorem{proposition}[theorem]{Proposition}
\newtheorem{lemma}[theorem]{Lemma}
\newtheorem{corollary}[theorem]{Corollary}
\theoremstyle{definition}
\newtheorem{definition}[theorem]{Definition}
\theoremstyle{remark}
\newtheorem{remark}[theorem]{Remark}

\newcommand{\Spine}{\mathrm{Spin}^{e}_{1,3}}

\newcommand{\SpineH}{\mathrm{Spin}^{e}_{H}}
\newcommand{\SOeH}{\mathrm{SO}^{e}_{H}}
\newcommand{\Cl}{\mathcal{C}\ell}
\newcommand{\End}{\operatorname{End}}
\newcommand{\HV}{H_V}
\newcommand{\HM}{H(M)}
\newcommand{\Dring}{\mathring D}
\newcommand{\gring}{{\mathring g}}
\newcommand{\wtens}{\,\widehat{\otimes}\,}
\newcommand{\RH}{\mathbb R^{H}}
\newcommand{\CRH}{\mathbb C\otimes\mathbb R^{H}}
\newcommand{\fH}{f_{H}}
\newcommand{\fcH}{\check f_{H}}
\newcommand{\IH}{I_H(M)}
\newcommand{\IcH}{I_{c,H}(M)}
\newcommand{\Pcan}{P^{\,\mathrm{can}}_{\SpineH}(M)}
\newcommand{\Ext}{{\textstyle\bigwedge}}

\begin{document}

\title{\textbf{Spin-Clifford Bundles and Spin Connections \\ on the Hyperbolic Clifford Algebra} \\ Part I}

\author{E. Notte-Cuello} 
\affil{Departamento de Matem\'aticas, Universidad de La Serena, Av. J. Cisternas 1200, Chile}

\author{J. M. Hoff da Silva} 
\affil{Departamento de F\'isica, Universidade Estadual Paulista, UNESP, Av. Dr. Ariberto Pereira da Cunha, 333, Guaratinguet\'a, SP, Brazil}

\date{}

\maketitle

\begin{abstract}
\noindent We develop the theory of spinor fields associated with the hyperbolic Clifford algebra $\RH=\Cl(\HV)$, which is associated with hyperbolic space  $\HV=V\oplus V^{*}$. Algebraically, we exhibit an explicit real primitive idempotent which, unlike the Lorentzian case, remains primitive after complexification and generates the ideal of Witten superfields. We also describe the chirality decomposition and the hyperbolic spin group. Geometrically, we construct a canonical hyperbolic spin structure on every oriented manifold, showing that no topological obstructions arise beyond orientability and we study the associated spinor bundles, spin-Clifford bundles and ideal spinor subbundles. Finally, we introduce the isotropic extension of an arbitrary linear connection to $TM\oplus T^*M$, which is always compatible with the hyperbolic metric, and the associated hyperbolic spin connection, whose well-definedness, compatibilities, curvature and explicit form are established, providing the geometric framework for Dirac operators on the hyperbolic Clifford algebra, which we will present in a forthcoming article.
\end{abstract}



\section{Introduction}

Clifford algebras and Clifford bundles provide a natural language for the geometry of spinor fields. In the Lorentzian setting, the systematic construction of the spin-Clifford bundles, of the bundles of algebraic spinor fields and of Dirac--Hestenes spinor fields as sections of appropriate bundles associated with the spin frame bundle was carried out by Mosna and Rodrigues \cite{MosnaRodrigues2004} and by Rodrigues \cite{Rodrigues2004}, and is presented in book form in \cite{RodriguesCapelas2016}. In that approach the Clifford bundle of differential forms $\Cl(M,g)$, the left and right spin-Clifford bundles $\Cl^{l}_{\Spine}(M,g)$, $\Cl^{r}_{\Spine}(M,g)$ and their ideal subbundles are carefully distinguished and the relations between them are made explicit. All these distinctions are the basis on which the different Dirac operators of the theory, and the corresponding field equations, are formulated.

The hyperbolic Clifford algebra $\Cl(\HV)$ of the hyperbolic space $\HV=V\oplus V^*$, endowed with the neutral bilinear form induced by the duality pairing, was studied by Rodrigues and de Souza \cite{RodriguesSouza2007} as the algebra of \emph{multivecfors}. It is often called the \emph{mother algebra}, since for every non-degenerate symmetric bilinear form $b$ on $V$ it is isomorphic to the graded tensor product $\Cl(V,b)\wtens\Cl(V,-b)$ and hence contains both the multivector and the multiform algebras of $(V,b)$. Moreover, it is also isomorphic to $\End(\Ext V)$ and its left minimal ideal $\RH\theta^*\simeq\Ext V^*$ consists of Witten superfields \cite{Witten,RodriguesSouza2007}. The corresponding bundle theory (the hyperbolic frame bundle, the hyperbolic Clifford bundle $\Cl(\HM)$, the hyperbolic Clifford--Lipschitz, Pin and Spin groups and the hyperbolic spin structures) was developed in \cite{NotteCuello2026}, where the question of possible topological obstructions to the existence of hyperbolic spin structures was left open. The answer, namely that there are none, follows from Karoubi's characteristic-class criterion for $\mathrm{Spin}(p,q)$-structures \cite[\S1.1, pp.~173--175]{Karoubi1968}, as observed for $TM\oplus T^*M$ in generalized geometry \cite[\S\S2.3 and 2.8, pp.~10--17]{Gualtieri2004}. This result was also shown in Ref. \cite{nov} by more elementary means. 

The purpose of the present paper is to provide in a self-contained way the whole algebraic and geometric scaffolding needed to formulate Dirac operators on the hyperbolic Clifford algebra. The paper is organized as follows. Section~\ref{sec:algebra} recalls the hyperbolic space $\HV$, the isotropic extension $A\mapsto A\oplus(A^{-1})^*$ of $\mathrm{GL}(V)$, the hyperbolic Clifford algebra $\RH$ and its structure as the mother algebra $\RH\simeq M_{2^n}(\mathbb R)$. Section~\ref{sec:alg-spinors} studies algebraic hyperbolic spinors. We give an explicit primitive idempotent $\fH=\prod_k\tfrac12 e_k\theta^k=\prod_k\tfrac12(1+\sigma_k\sigma_{n+k})$ of $\RH$. In contrast with the Lorentzian idempotent $\tfrac12(1+\mathbf E^0)\tfrac12(1+i\mathbf E^2\mathbf E^1)$, $\fH$ is real and remains primitive in $\CRH$. It generates the ideal $\RH\fH=\RH\theta^*\simeq\Ext V^*$; note that the volume element $\theta^*$ of $V^*$, which is sometimes regarded as the generator of this ideal, is nilpotent and not idempotent. We also determine all primitive idempotents, the chirality decomposition, the commutant of the spinor module. The invariant spinor pairing is introduced in Section~\ref{sec:groups}, which also recalls the hyperbolic Clifford--Lipschitz, Pin and Spin groups of \cite{NotteCuello2026} and adds the Lie-theoretic results needed later. We identify the Lie algebra of $\SpineH$ with $\Ext^2\HV$ and construct a canonical lift $\tilde\jmath:\mathrm{GL}^+(V)\to\SpineH$ of the isotropic extension $A\mapsto A\oplus(A^{-1})^*$. The invariance of the spinor pairing under $\SpineH$ is also established here\footnote{All the finite-dimensional algebraic identities of Sections~\ref{sec:alg-spinors} and~\ref{sec:groups} (Witt relations, idempotency and primitivity of $\fH$, the ideal identifications, the chirality formula, the lift of $\mathrm{GL}^+(V)$, and the spinor pairing and its signature) have been checked by direct computer algebra in $\Cl(\HV)\simeq M_{16}(\mathbb R)$ for $n=4$.}. Section~\ref{sec:spinstructure} deals with the hyperbolic Clifford bundle $\Cl(\HM)$, its canonical global sections and hyperbolic spin structures. We show that every oriented manifold carries a canonical hyperbolic spin structure $\Pcan$, giving an explicit construction, in the present formalism, of the spin structures whose existence follows from Karoubi's criterion \cite[\S1.1, pp.~173--175]{Karoubi1968} (see also \cite[\S\S2.3 and 2.8, pp.~10--17]{Gualtieri2004}), and any other hyperbolic spin structure differs from it by a class in $H^1(M;\mathbb Z_2)$.

In Section~\ref{sec:bundles} we study the hyperbolic spinor and spin-Clifford bundles and their ideal subbundles. We define the left and right hyperbolic spinor bundles (real and complex), the left and right hyperbolic spin-Clifford bundles, prove that they are module bundles over $\Cl(\HM)$ and describe their ideal spinor subbundles. For the canonical spin structure these are realized as $\Ext T^*M\otimes(\Ext^nTM)^{1/2}$, in agreement with the description of spinors of $TM\oplus T^*M$ in generalized geometry \cite{Gualtieri2004}. In Section~\ref{sec:connections} we show that the isotropic extension $\nabla^H$ of an \emph{arbitrary} linear connection $\nabla$ on $M$ is always compatible with the hyperbolic metric and with the Clifford product; the canonical idempotent $\fH$ and the volume element of $\HM$ are parallel for it. We define the hyperbolic spin connection induced by $\nabla^H$ and establish its well-definedness, its compatibility with all the structures above, its curvature and its explicit form. Finally, Section~\ref{sec:conclusions} contains the conclusions.

\section{The hyperbolic Clifford algebra}\label{sec:algebra}

We start by recalling the basic mathematical tools related to the hyperbolic Clifford algebra. Let $V$ be a real vector space of dimension\footnote{The assumption $n\geq3$ is used only to have $\pi_1(\mathrm{GL}^+(V))\simeq\mathbb Z_2$ in the proof of Proposition~\ref{prop:lift}. For $n=2$ one has $\pi_1(\mathrm{GL}^+(V))\simeq\mathbb Z$, generated by the same rotation, and the computation $\exp(2\pi X_a)=1$ in that proof shows that the generator still lifts to a closed loop. Hence Proposition~\ref{prop:lift} and Theorem~\ref{thm:canonical} hold for $n=2$ as well (for $n=1$, $\mathrm{GL}^+(V)$ is contractible and the lift exists trivially).} $n\geq 3$ (the case of physical interest being $V=\mathbb R^{1,3}$, $n=4$), possibly equipped with a non-degenerate symmetric bilinear form $b$. Following \cite[\S 2.1]{NotteCuello2026}, we denote by
$$\HV=(V\oplus V^{*},\langle\,,\,\rangle)$$
the \emph{hyperbolic space} associated with $V$, with the non-degenerate symmetric bilinear form
\begin{equation}\label{eq:hypform}
\langle \mathbf x,\mathbf y\rangle=x^{*}(y_{*})+y^{*}(x_{*}),\qquad \mathbf x=x_*\oplus x^{*},\ \mathbf y=y_*\oplus y^{*}\in \HV,
\end{equation}
which has neutral signature $(n,n)$ and does not depend on $b$. The elements of $\HV$ are called \emph{vecfors}. The subspaces $V$ and $V^{*}$, identified in $\HV$ by $x_*\mapsto x_*\oplus 0$ and $x^{*}\mapsto 0\oplus x^{*}$, are maximal totally isotropic subspaces of $\HV$ \cite[\S 2.1]{NotteCuello2026}. Given a basis $\{e_1,\dots,e_n\}$ of $V$ and its dual basis $\{\theta^1,\dots,\theta^n\}$ of $V^{*}$, the basis $\{e_1,\dots,e_n,\theta^1,\dots,\theta^n\}$ of $\HV$ is a \emph{Witt basis}, and the matrix of $\langle\,,\,\rangle$ in it is
$$\begin{pmatrix}0 & I_n\\ I_n & 0\end{pmatrix},$$
which explicitly exhibits the neutrality of the form \cite[Rem.\ 1]{NotteCuello2026}. An orthonormal basis $\{\sigma_1,\dots,\sigma_{2n}\}$ of $\HV$ is obtained from the Witt basis by
\begin{equation}\label{eq:orthonormal}
\sigma_k=\tfrac{1}{\sqrt2}(e_k\oplus\theta^k),\qquad \sigma_{n+k}=\tfrac{1}{\sqrt2}(-e_k\oplus\theta^k),\qquad k=1,\dots,n,
\end{equation}
satisfying $\langle\sigma_k,\sigma_l\rangle=\delta_{kl}$, $\langle\sigma_{n+k},\sigma_{n+l}\rangle=-\delta_{kl}$, $\langle\sigma_k,\sigma_{n+l}\rangle=0$ \cite{RodriguesSouza2007}. Conversely,
\begin{equation}\label{eq:witt-from-orth}
e_k=\tfrac1{\sqrt2}(\sigma_k-\sigma_{n+k}),\qquad \theta^k=\tfrac1{\sqrt2}(\sigma_k+\sigma_{n+k}).
\end{equation}

For $A\in\mathrm{GL}(V)$ we put
\begin{equation}\label{eq:j}
j(A):=A\oplus(A^{-1})^*\in\mathrm{GL}(\HV),
\end{equation}
the isotropic extension\footnote{Explicitly, for $x_*\in V$ and $x^*\in V^*$,
\[
  j(A)(x_*\oplus x^*)=Ax_*\,\oplus\,x^*\circ A^{-1}:
\]
$A$ acts on vectors, and the \emph{inverse transpose} $(A^{-1})^*:x^*\mapsto x^*\circ A^{-1}$ acts on covectors. In the Witt basis $\{e_1,\dots,e_n,\theta^1,\dots,\theta^n\}$ the matrix of $j(A)$ is block diagonal,
\[
  [j(A)]=\begin{pmatrix}[A]&0\\[2pt]0&([A]^{-1})^{T}\end{pmatrix},
\]
where $[A]$ is the matrix of $A$ in the basis $\{e_k\}$.} of $A$ (cf.\ the isotropic extension $I(\phi)=\phi\oplus\phi^*$ of an endomorphism in \cite[\S 3.1]{RodriguesSouza2007}). Since $\langle\,,\,\rangle$ is built from the duality pairing only, $j(A)$ is an isometry of $\HV$; $j$ is a group homomorphism, $\det j(A)=1$, and $j$ maps the Witt basis associated with $\{e_k\}$ to the Witt basis associated with $\{Ae_k\}$. Its differential is $j_*(a)=a\oplus(-a^*)$, $a\in\mathfrak{gl}(V)$. Since $\mathrm{GL}^+(V)$ is connected, $j(\mathrm{GL}^+(V))$ is contained in the identity component $\SOeH$ of the orthogonal group $O_H$ of $\HV$.

The hyperbolic Clifford algebra is
$$\RH:=\Cl(\HV)=\Cl(\HV,\langle\,,\,\rangle),$$
of dimension $2^{2n}$ (for $n=4$, $\dim\RH=2^{8}=256$, versus $\dim \mathbb R_{1,3}=2^4=16$ in the classical case). The notation $\mathbb R^{H}_{1,3}$ is also used in the case $V=\mathbb R^{1,3}$. Since $\langle\,,\,\rangle$ does not depend on $b$, all statements below hold for arbitrary $n$. The fundamental anticommutation relation is
\begin{equation}\label{eq:anticomm}
\mathbf x\mathbf y+\mathbf y\mathbf x=2\langle\mathbf x,\mathbf y\rangle,\qquad \mathbf x,\mathbf y\in \HV,
\end{equation}
analogous to $\theta^a\theta^b+\theta^b\theta^a=2\eta^{ab}$ in the Lorentzian case. The Clifford product of $\mathbf x\in\HV$ by $u\in\Ext\HV$ is $\mathbf xu=\mathbf x\lrcorner u+\mathbf x\wedge u$, with $\lrcorner$ the left contraction induced by $\langle\,,\,\rangle$ \cite{RodriguesSouza2007}. We denote by $\hat u$ the grade involution, by $\tilde u$ the reversion and by $\bar u=\hat{\tilde u}$ the conjugation of $u\in\RH$; $(\RH)^{(0)}$ is the even subalgebra, and $\mathbb R^{H\star}$ the group of invertible elements.

In a Witt basis, \eqref{eq:anticomm} reads
\begin{equation}\label{eq:witt}
e_ke_l=-e_le_k,\qquad \theta^k\theta^l=-\theta^l\theta^k,\qquad e_k\theta^l+\theta^le_k=2\delta_k^l .
\end{equation}
In particular the Clifford product of elements of $V$ (resp.\ $V^*$) coincides with their exterior product, and $\Ext V$, $\Ext V^*$ are subalgebras of $\RH$. We write
$$e_*:=e_1e_2\cdots e_n=e_1\wedge\cdots\wedge e_n,\qquad \theta^*:=\theta^1\theta^2\cdots\theta^n=\theta^1\wedge\cdots\wedge\theta^n .$$
Note that $(e_*)^2=0=(\theta^*)^2$ by \eqref{eq:witt}. Every $x\in\RH$ is a linear combination of the monomials $\theta^Je_I$, $I,J\subset\{1,\dots,n\}$ ordered. For $x_*\in V$ we denote by $i_{x_*}$ the interior product on $\Ext V^*$. For $\gamma\in\Ext V^*$ we write $\gamma=\sum_{r=0}^n\gamma_{[r]}$, with $\gamma_{[r]}\in\bigwedge^rV^*$ its homogeneous component of degree $r$; in particular $\gamma_{[n]}\in\bigwedge^nV^*=\mathbb R\theta^*$.

Finally, the \emph{volume element} of $\HV$ is
\begin{equation}\label{eq:volume}
\varpi:=\sigma_1\sigma_2\cdots\sigma_{2n}\in\Ext^{2n}\HV .
\end{equation}
It does not depend on the choice of the basis $\{e_k\}$ of $V$ used in \eqref{eq:orthonormal}: a change of basis $e\mapsto eA$ changes the orthonormal basis by $j(A)$, and $\det j(A)=1$. Its basic properties are given in Proposition~\ref{prop:chir}.

With these definitions, we review the following result of Rodrigues and de Souza, which is the structural fact behind the name ``mother algebra''.

\begin{proposition}[\cite{RodriguesSouza2007}]\label{prop:tensor-iso}
There is a natural isomorphism $\Cl(\HV)\simeq\End(\Ext V)$ and, for every non-degenerate symmetric bilinear form $b$ on $V$,
\begin{equation}\label{eq:tensor-iso}
\Cl(\HV)\;\simeq\;\Cl(H_{bV})\;\simeq\;\Cl(V,b)\wtens\Cl(V,-b),
\end{equation}
induced by the Clifford map $x_*\oplus x^*\mapsto x_+\wtens 1+1\wtens x_-$, with $x_\pm=\tfrac1{\sqrt2}(b^*x^*\pm x_*)$, where $\wtens$ denotes the graded tensor product and $b^*:V^*\to V$ is the isomorphism induced by $b$.
\end{proposition}

Thus $\Cl(\HV)$ simultaneously contains $\Cl(V,b)$ and $\Cl(V,-b)$, and allows multivectors and multiforms to be represented simultaneously. The first isomorphism in Proposition~\ref{prop:tensor-iso} gives an explicit matrix realization of $\RH$ and of its complexification. This is the hyperbolic counterpart of the identification $\mathbb C\otimes\mathbb R_{1,3}\simeq\mathbb R_{4,1}\simeq\mathbb C(4)$ used for complex spinor bundles in the Lorentzian case \cite[Def.\ 13]{MosnaRodrigues2004}.

\begin{proposition}[Matrix realization and complexification]\label{prop:realization}
$\RH$ is isomorphic, as a real algebra, to the full matrix algebra $M_{2^n}(\mathbb R)$ (for $n=4$, $M_{16}(\mathbb R)$) and its complexification is
$$\CRH\;\simeq\;M_{2^n}(\mathbb C)\qquad(n=4:\ M_{16}(\mathbb C)).$$
\end{proposition}

\begin{proof}
By Proposition~\ref{prop:tensor-iso}, $\Cl(\HV)\simeq\End(\Ext V)$. Since $\dim_{\mathbb R} V=n$, $\Ext V$ has real dimension $2^n$, so $\End(\Ext V)\simeq M_{2^n}(\mathbb R)$, of real dimension $2^{2n}=\dim_{\mathbb R}\RH$. This reflects the fact that $\HV$ carries a hyperbolic (maximally split) bilinear form: the Clifford algebra of a neutral form on a $2n$-dimensional space is the full real matrix algebra of rank $2^n$, the maximal isotropic subspace $V$ providing the carrier space $\Ext V$ of a simple module. Complexifying, and using that extension of scalars commutes with taking endomorphism algebras of finite-dimensional vector spaces,
$$\CRH\;\simeq\;\mathbb C\otimes\End(\Ext V)\;\simeq\;\End_{\mathbb C}(\mathbb C\otimes\Ext V)\;\simeq\;M_{2^n}(\mathbb C).$$ 
\end{proof}

\section{Algebraic hyperbolic spinors}\label{sec:alg-spinors}

In the Lorentzian algebra $\mathbb R_{1,3}$ the spinor spaces are left minimal ideals $\mathbb R_{1,3}\tfrac12(1+\mathbf E^0)$ generated by primitive idempotents \cite[Eqs.\ 17--19]{MosnaRodrigues2004}. In this section we describe the corresponding objects in $\RH$. It is known that the dual exterior algebra is isomorphic, as a left $\RH$-module, to a left minimal ideal of $\RH$, $\RH\,\theta^{*}=\Ext V^*\fH\simeq\Ext V^*$ \cite[\S 6]{RodriguesSouza2007}. Its elements are Witten superfields \cite{Witten}. Since $(\theta^*)^2=0$, however, $\theta^*$ is nilpotent and cannot be the idempotent generating this ideal. The next theorem exhibits the correct generator.

\begin{theorem}[Explicit primitive idempotent]\label{thm:fH}
Let $\{e_k,\theta^k\}$ be a Witt basis, $f_k:=\tfrac12e_k\theta^k$ $(k=1,\dots,n)$ and
$$\fH:=f_1f_2\cdots f_n .$$
Then:
\begin{enumerate}[label=(\roman*)]
\item the $f_k$ are pairwise commuting idempotents, hence $\fH^2=\fH$;
\item $\displaystyle \fH=2^{-n}e_*\tilde\theta^*=2^{-n}(-1)^{n(n-1)/2}e_*\theta^*=\prod_{k=1}^n\tfrac12\big(1+\sigma_k\sigma_{n+k}\big)$, with $(\sigma_k\sigma_{n+k})^2=1$ and $[\sigma_k\sigma_{n+k},\sigma_l\sigma_{n+l}]=0$; moreover $\fH$ does not depend on the choice of the basis of $V$: $j(A)_\#\fH=\fH$ for all $A\in\mathrm{GL}(V)$, where $j(A)_\#$ is the automorphism of $\RH$ extending $j(A)$;
\item $e_k\fH=0=\fH\theta^k$ for all $k$, $\theta^*\fH=\theta^*$ and $\fH e_*=e_*$;
\item $\fH\,\RH\,\fH=\mathbb R\fH$ and $\fH(\CRH)\fH=\mathbb C\fH$; hence $\fH$ is a primitive idempotent of $\RH$ \emph{and} of $\CRH$;
\item $\RH\fH=\RH\theta^*=\{\alpha\fH:\alpha\in\Ext V^*\}$, the map $\Ext V^*\to\RH\fH$, $\alpha\mapsto\alpha\fH$, is a linear isomorphism, and
\begin{equation}\label{eq:fock}
\theta^k(\alpha\fH)=(\theta^k\wedge\alpha)\fH,\qquad e_k(\alpha\fH)=2\,(i_{e_k}\alpha)\fH ;
\end{equation}
\item $\fH\RH=e_*\RH=\{\fH\lambda:\lambda\in\Ext V\}$, and $\lambda\mapsto\fH\lambda$ is a linear isomorphism $\Ext V\to\fH\RH$;
\item $\tilde f_H=\fcH:=\prod_k\tfrac12\theta^ke_k$, $f_k+\tfrac12\theta^ke_k=1$, and for every $\gamma\in\Ext V^*$, $\fcH\gamma\fH=\gamma_{[n]}$, the component of $\gamma$ in $\Ext^nV^*=\mathbb R\theta^*$. Consequently $\fcH x\fH\in\mathbb R\theta^*$ for all $x\in\RH$.
\end{enumerate}
\end{theorem}

\begin{proof}
(i) By \eqref{eq:witt}, $(e_k\theta^k)^2=e_k(2-e_k\theta^k)\theta^k=2e_k\theta^k$ since $e_k^2=0$, so $f_k^2=f_k$. For $l\neq k$, $e_l$ and $\theta^l$ anticommute with both $e_k$ and $\theta^k$, hence commute with the even element $e_k\theta^k$; thus $f_kf_l=f_lf_k$.

(ii) By induction, $\prod_{k=1}^n e_k\theta^k=e_1\cdots e_n\theta^n\cdots\theta^1$. Passing from $n-1$ to $n$, $e_n$ is moved to the left across $n-1$ factors $\theta^k$ and $\theta^n$ to the left of $\theta^{n-1}\cdots\theta^1$, each move producing $(-1)^{n-1}$. Since $\theta^n\cdots\theta^1=\tilde\theta^*=(-1)^{n(n-1)/2}\theta^*$, the first two expressions follow. From \eqref{eq:witt-from-orth}, $e_k\theta^k=\tfrac12(\sigma_k-\sigma_{n+k})(\sigma_k+\sigma_{n+k})=\tfrac12(\sigma_k^2-\sigma_{n+k}^2)+\sigma_k\sigma_{n+k}=1+\sigma_k\sigma_{n+k}$, and $(\sigma_k\sigma_{n+k})^2=-\sigma_k^2\sigma_{n+k}^2=1$; bivectors with disjoint support commute. Finally $j(A)_\#$ maps $e_*\mapsto(\det A)e_*$ and $\theta^*\mapsto(\det A)^{-1}\theta^*$, so $j(A)_\#\fH=\fH$.

(iii) $e_kf_k=\tfrac12e_k^2\theta^k=0$ and $f_k\theta^k=0$; as $f_l$ ($l\ne k$) commutes with $e_k,\theta^k$, $e_k\fH=0=\fH\theta^k$. By \eqref{eq:witt}, $\theta^k e_k\theta^k=(2-e_k\theta^k)\theta^k=2\theta^k$, so $\theta^kf_k=\theta^k$, and similarly $f_ke_k=e_k$. Since $f_k$ commutes with $\theta^l,e_l$ for $l\neq k$, $\theta^*\fH=\theta^*$ and $\fH e_*=e_*$.

(iv) Expand $x\in\RH$ in the monomials $\theta^Je_I$. By (iii), $e_I\fH=0$ if $I\neq\emptyset$ and $\fH\theta^J=0$ if $J\neq\emptyset$, so $\fH x\fH=c_{\emptyset\emptyset}\fH$. An idempotent $f$ of an algebra $A$ is primitive iff $fAf$ is a division algebra, here it is $\mathbb R$ (resp.\ $\mathbb C$, since $\fH(\CRH)\fH=\mathbb C\otimes\fH\RH\fH$).

(v) By the same monomial expansion, $x\fH=\sum_J c_{J\emptyset}\theta^J\fH$, so $\RH\fH=\Ext V^*\fH$. For Injectivity, if $\alpha=\sum_Jc_J\theta^J\neq0$, let $J_0$ be of minimal cardinality with $c_{J_0}\neq0$ and $K$ its complement. Then $\theta^K\alpha=\pm c_{J_0}\theta^*$ (all other terms vanish by minimality), whence $\theta^K\alpha\fH=\pm c_{J_0}\theta^*\neq0$ by (iii). Next, $\theta^*=\theta^*\fH\in\RH\fH$ and, by (ii), $\fH\propto e_*\theta^*\in\RH\theta^*$, so $\RH\fH=\RH\theta^*$. The first formula \eqref{eq:fock} is the identity $\theta^k\alpha=\theta^k\wedge\alpha$ in $\Ext V^*$. For the second, by \eqref{eq:witt}, $e_k\theta^{j_1}\cdots\theta^{j_r}\fH=\sum_{s}(-1)^{s-1}2\delta^{j_s}_k\,\theta^{j_1}\cdots\widehat{\theta^{j_s}}\cdots\theta^{j_r}\fH+(-1)^r\theta^{j_1}\cdots\theta^{j_r}e_k\fH=2(i_{e_k}\theta^J)\fH$.

(vi) The isometry $\tau$ of $\HV$ exchanging $e_k\leftrightarrow\theta^k$ extends to an automorphism $\tau_\#$ of $\RH$ with $\tau_\#(\fH)=\fcH$. The anti-automorphism $x\mapsto\tau_\#(\tilde x)$ fixes $\fH$ (by (vii) below), exchanges $\Ext V^*$ and $\Ext V$, and maps $\RH\fH$ onto $\fH\RH$; (vi) follows from (v). Also $e_*\RH=\fH\RH$ since $e_*=\fH e_*$ and both are minimal right ideals.

(vii) Reversion reverses the order of the commuting factors and $\widetilde{e_k\theta^k}=\theta^ke_k$, thus $f_k+\tfrac12\theta^ke_k=\tfrac12(e_k\theta^k+\theta^ke_k)=1$. Put $\check f_k=\tfrac12\theta^ke_k$, then $\check f_kf_k=\tfrac14\theta^ke_ke_k\theta^k=0$, $\check f_k\theta^k=\theta^k=\theta^kf_k$ (as in (iii)), and $\check f_k,f_k$ commute with $\theta^l$, $l\neq k$. Hence for $\gamma=\theta^J$: if some $k\notin J$, the factors $\check f_k$ and $f_k$ can be brought together and $\fcH\theta^J\fH=0$; if $J=\{1,\dots,n\}$, $\fcH\theta^*\fH=\theta^*$. The last statement follows from (v).
\end{proof}

By Theorem~\ref{thm:fH}(ii), $\fH$ is a canonical element of $\RH$ determined by the decomposition $\HV=V\oplus V^*$ alone. We call it the \emph{canonical primitive idempotent} of $\RH$ and $S:=\RH\fH\simeq\Ext V^*$ the \emph{canonical spinor module}. Formulas \eqref{eq:fock} show that, on $S\simeq\Ext V^*$, covectors act by exterior multiplication and vectors by (twice) interior multiplication: $S$ is the Fock representation of the canonical anticommutation relations \eqref{eq:witt}, with vacuum $\fH$.

\begin{remark}[Comparison with the Lorentzian idempotent]\label{rem:comparison}
In the Lorentzian case $\mathbb R_{1,3}\simeq M_2(\mathbb H)$ is not split, the real primitive idempotent $\tfrac12(1+\mathbf E^0)$ ceases to be primitive after complexification ($\mathbb C\otimes\mathbb R_{1,3}\simeq\mathbb C(4)$) and a second factor $\tfrac12(1+i\mathbf E^2\mathbf E^1)$ containing $i$ is required because $(\mathbf E^2\mathbf E^1)^2=-1$. This gives the idempotent $\mathbf f=\tfrac12(1+\mathbf E^0)\tfrac12(1+i\mathbf E^2\mathbf E^1)$ of \cite[Eq.~(18)]{MosnaRodrigues2004}. In the hyperbolic case every factor $\tfrac12(1+\sigma_k\sigma_{n+k})$ is built from a bivector of square $+1$ (one positive and one negative vector of the neutral space), $\RH\simeq M_{2^n}(\mathbb R)$ is already split, and the \emph{real} idempotent $\fH$ remains primitive in $\CRH$ (Theorem~\ref{thm:fH}(iv)): no imaginary unit is needed. Thus $\fH$ is the hyperbolic analogue of $\mathbf f$, and complex hyperbolic algebraic spinors are simply the complexification of the real ones (see Theorem~\ref{thm:ideal}(v) below).
\end{remark}

Now we establish results concerning all primitive idempotents, chirality and the commutant of the spinor module in sequence.  

\begin{proposition}[Primitive idempotents]\label{prop:idempotents}
\begin{enumerate}[label=(\alph*)]
\item Every primitive idempotent of $\RH$ (resp.\ $\CRH$) is of the form $g\fH g^{-1}$ with $g$ invertible in $\RH$ (resp.\ $\CRH$).
\item For $J\subset\{1,\dots,n\}$ put $p_J:=\prod_{k\notin J}f_k\prod_{k\in J}\check f_k$. The $2^n$ elements $p_J$ are primitive idempotents, $p_Jp_{J'}=\delta_{JJ'}p_J$ and $\sum_Jp_J=1$, hence $\RH=\bigoplus_J\RH p_J$, with $p_\emptyset=\fH$ and $p_{\{1,\dots,n\}}=\fcH$.
\item $f_c:=\tfrac12(1+i\sigma_1\sigma_2)\,\tfrac12(1+i\sigma_{n+1}\sigma_{n+2})\prod_{k=3}^n\tfrac12(1+\sigma_k\sigma_{n+k})$ is a non-real primitive idempotent of $\CRH$, of the same shape as $\mathbf f$; by (a) it is conjugate to $\fH$.
\end{enumerate}
\end{proposition}
\begin{proof}
(a) Under $\RH\simeq M_N(\mathbb R)$, $N=2^n$ (Proposition~\ref{prop:realization}), and $\CRH\simeq M_N(\mathbb C)$, primitive idempotents are exactly the idempotents of rank one and any two of them are conjugate (choose bases adapted to image and kernel).
(b) With $\tau_J$ the isometry of $\HV$ exchanging $e_k\leftrightarrow\theta^k$ for $k\in J$, $p_J=(\tau_J)_\#\fH$ is primitive. Orthogonality follows from $f_k\check f_k=\check f_kf_k=0$ and the commutativity of all factors. Also  $\sum_Jp_J=\prod_k(f_k+\check f_k)=1$ by Theorem~\ref{thm:fH}(vii).
(c) The factors commute (disjoint supports) and are idempotents since $(i\sigma_1\sigma_2)^2=(i\sigma_{n+1}\sigma_{n+2})^2=1$. The scalar part $\langle\ \rangle_0$ is the normalized trace of $M_N(\mathbb C)$ (it satisfies $\langle xy\rangle_0=\langle yx\rangle_0$ and $\langle1\rangle_0=1$). Expanding the product, all non-trivial products of the $n$ bivectors are blades of positive grade, so $\langle f_c\rangle_0=2^{-n}$ and $\operatorname{rank}f_c=N\langle f_c\rangle_0=1$.
\end{proof}

\begin{proposition}[Chirality]\label{prop:chir}
The volume element \eqref{eq:volume} satisfies $\varpi^2=1$, $\varpi x=\hat x\varpi$ for $x\in\RH$, and
$$\varpi=(-1)^{n(n-1)/2}\prod_{k=1}^n\sigma_k\sigma_{n+k}=(-1)^{n(n-1)/2}\prod_{k=1}^n(2f_k-1),\qquad \varpi(\alpha\fH)=(-1)^{n(n-1)/2}\,(\hat\alpha)\fH .$$
Hence $S=\RH\fH=S^+\oplus S^-$, with $S^+=\Ext^{\mathrm{ev}}V^*\fH$ and $S^-=\Ext^{\mathrm{odd}}V^*\fH$ the $\pm(-1)^{n(n-1)/2}$-eigenspaces of $\varpi$ and both are invariant under the even subalgebra $(\RH)^{(0)}$.
\end{proposition}
\begin{proof}
Reordering $\sigma_1\cdots\sigma_{2n}$ into $\prod_k\sigma_k\sigma_{n+k}$ requires $(n-1)+(n-2)+\dots+0=n(n-1)/2$ transpositions. By the proof of Theorem~\ref{thm:fH}(ii), $\sigma_k\sigma_{n+k}=e_k\theta^k-1=2f_k-1$. These commute and square to $1$, so $\varpi^2=1$. Since $2n$ is even, $\varpi$ anticommutes with each $\sigma_A$. If $k\in J$, $f_k\theta^J\fH=0$ (move $f_k$ to $\theta^k$, $f_k\theta^k=0$). If $k\notin J$, $f_k\theta^J\fH=\theta^Jf_k\fH=\theta^J\fH$. So $(2f_k-1)$ acts on $\theta^J\fH$ as $(-1)^{[k\in J]}$. Even elements commute with $\varpi$.
\end{proof}

\begin{proposition}[Commutant]\label{prop:commutant}
$\End_{\RH}(\RH\fH)\simeq(\fH\RH\fH)^{\mathrm{op}}=\mathbb R$ and $\End_{\CRH}((\CRH)\fH)\simeq(\fH(\CRH)\fH)^{\mathrm{op}}=\mathbb C$, acting by right multiplication by $\fH\RH\fH$ (resp.\ $\fH(\CRH)\fH$). In particular the real ideal $\RH\fH$ carries \emph{no} $\RH$-linear complex structure. By contrast, in the Lorentzian case $\tfrac12(1+\mathbf E^0)\mathbb R_{1,3}\tfrac12(1+\mathbf E^0)\simeq\mathbb H$, which is what allows $\mathbf E^{21}$ to act as a complex structure on the real ideal in \cite{MosnaRodrigues2004}.
\end{proposition}
\begin{proof}
Put $A=\RH$ (resp.\ $\CRH$) and $f=\fH$. If $\phi\in\End_A(Af)$, then $\phi(f)=\phi(f\cdot f)=f\,\phi(f)\in fAf$, and for every $af\in Af$, $\phi(af)=a\,\phi(f)=af\,c$ with $c:=\phi(f)\in fAf$. Thus $\phi$ is right multiplication $R_c$ by $c$, and conversely every $R_c$, $c\in fAf$, is $A$-linear. Since $R_cR_{c'}(x)=xc'c$, i.e.\ $R_cR_{c'}=R_{c'c}$, the map $c\mapsto R_c$ reverses the order of products. This is the reason for the superscript ``op'': for an algebra $B$, the \emph{opposite algebra} $B^{\mathrm{op}}$ is the same vector space as $B$, with the same unit, endowed with the reversed product $c\ast c':=c'c$. With respect to $\ast$ one has $R_{c\ast c'}=R_{c'c}=R_cR_{c'}$ and $c\mapsto R_c$ is injective since $R_c(f)=fc=c$, hence $c\mapsto R_c$ is an algebra isomorphism $(fAf)^{\mathrm{op}}\to\End_A(Af)$, i.e.\ an anti-isomorphism $fAf\to\End_A(Af)$. By Theorem~\ref{thm:fH}(iv), $fAf=\mathbb R f\simeq\mathbb R$ (resp.\ $\mathbb C f\simeq\mathbb C$), which is commutative, so that $(fAf)^{\mathrm{op}}=fAf$ and the reversal of the order plays no role here. An $\RH$-linear complex structure on $\RH\fH$ would be some $\mathcal J\in\mathbb R$ with $\mathcal J^2=-1$, which is impossible. Finally, in the Lorentzian case, instead, $fAf\simeq\mathbb H$ contains square roots of $-1$.
\end{proof}

\section{Hyperbolic Clifford--Lipschitz, Pin and Spin groups}\label{sec:groups}

This section starts recalling, with notation adapted to that of this work, the construction of the groups associated with $\Cl(H_{p,q})$ given in \cite[\S 3]{NotteCuello2026}. Here $H_{p,q}$ denotes the hyperbolic space of $V=\mathbb R^{p,q}$, $p+q=n$. Let $\mathbb R^{H\star}_{p,q}$ be the group of invertible elements of $\mathbb R^{H}_{p,q}=\Cl(H_{p,q})$. The twisted adjoint action $\rho_{\mathbf x}(v)=\hat{\mathbf x}\,v\,\mathbf x^{-1}$ allows one to define the hyperbolic \emph{Clifford--Lipschitz group}
\begin{equation}\label{eq:CL}
\Gamma_{H_{p,q}}=\{\mathbf x\in \mathbb R^{H\star}_{p,q}\;:\;\hat{\mathbf x}\,\mathbf v\,\mathbf x^{-1}\in H_{p,q},\ \forall \mathbf v\in H_{p,q}\},
\end{equation}
with action $\rho:\Gamma_{H_{p,q}}\to \mathrm{Aut}(H_{p,q})$, $\rho_{\mathbf x}\mathbf v=\hat{\mathbf x}\mathbf v\mathbf x^{-1}$, which preserves the quadratic form of $H_{p,q}$ \cite[Eqs.\ (18)--(19)]{NotteCuello2026}.

\begin{theorem}[\cite{NotteCuello2026}]\label{thm:rho}
The map $\rho:\Gamma_{H_{p,q}}\to O_{H_{p,q}}$ is a surjective homomorphism whose kernel is the multiplicative group $\mathbb R^{*}\mathbf 1$ of nonzero scalar multiples of the unit of $\Cl(H_{p,q})$.
\end{theorem}

Restricting $\Gamma_{H_{p,q}}$ to elements of norm $N(\mathbf x)=\mathbf x\bar{\mathbf x}=\pm1$ gives the \emph{hyperbolic Pin group},
$$\mathrm{Pin}_{H_{p,q}}=\{\mathbf x\in\Gamma_{H_{p,q}}\;:\;N(\mathbf x)=\pm1\}$$
and further restricting to the even part, $\Gamma^{+}_{H_{p,q}}=\Gamma_{H_{p,q}}\cap\Cl^{(0)}(H_{p,q})$, gives the \emph{hyperbolic Spin group},
$$\mathrm{Spin}_{H_{p,q}}=\{\mathbf x\in\Gamma^{+}_{H_{p,q}}\;:\;N(\mathbf x)=\pm1\}.$$
One has $\ker(\rho|_{\mathrm{Pin}_{H_{p,q}}})=\ker(\rho|_{\mathrm{Spin}_{H_{p,q}}})=\{\pm1\}\simeq\mathbb Z_2$ \cite[\S 3.3]{NotteCuello2026}, whence
\begin{equation}\label{eq:quotients}
O_{H_{p,q}}=\frac{\mathrm{Pin}_{H_{p,q}}}{\mathbb Z_2},\qquad SO_{H_{p,q}}=\frac{\mathrm{Spin}_{H_{p,q}}}{\mathbb Z_2},\qquad SO^e_{H_{p,q}}=\frac{\mathrm{Spin}^{e}_{H_{p,q}}}{\mathbb Z_2},
\end{equation}
where $\mathrm{Spin}^{e}_{H_{p,q}}$ denotes the identity component of $\mathrm{Spin}_{H_{p,q}}$.

\begin{remark}
As emphasized in \cite{NotteCuello2026}, the structural difference from the classical case does not lie in the form of these constructions, formally identical to those of $\Cl(V,g)$, but rather in the dimension of the space on which they act: whereas the classical $\Gamma_{p,q}$ acts on an $n$-dimensional space, $\Gamma_{H_{p,q}}$ acts on the $2n$-dimensional space $H_{p,q}$, which allows multivectors and multiforms to be represented simultaneously.
\end{remark}

Since $\langle\,,\,\rangle$ does not depend on $(p,q)$, the groups above depend only on $n$. We write $\mathrm{Spin}_{H}$ and $\SpineH$ (or $\mathrm{Spin}^{e}_{H_{1,3}}$ when $V=\mathbb R^{1,3}$) for the hyperbolic Spin group and its identity component and $\SOeH$ for the identity component of $O_H$. For $u\in\mathrm{Spin}_H$, which is even, $\rho_u(v)=uvu^{-1}$, and its extension to $\RH$ is the inner automorphism $\mathrm{Ad}_u$.

\begin{lemma}[Lie algebra and normalization]\label{lem:lie}
$\Ext^2\HV\subset\RH$ is closed under commutators and is the Lie algebra of\footnote{The statement emulates the classical case of the Clifford algebra of multivectors (or multiforms) $\mathbb R_{1,3}=\Cl(\mathbb R^{1,3},\eta)$, where every $u\in\Spine$ can be written as $u=\pm e^{F}$ with $F\in\Ext^2\mathbb R^{1,3}\hookrightarrow\mathbb R_{1,3}$, so that the Lie algebra of $\Spine$ is the space of bivectors $\Ext^2\mathbb R^{1,3}$ with the commutator product \cite[\S 3.3.4]{RodriguesCapelas2016}. There the base space is $\mathbb R^{1,3}$ because it is the space that generates the Clifford algebra and on which $\Spine$ acts by $\mathrm{Ad}$. In $\RH=\Cl(\HV)$ both roles are played by the whole hyperbolic space $\HV=V\oplus V^*$: it generates $\RH$ and $\SpineH$ acts on it through $\rho$. Hence the analogue of $\Ext^2\mathbb R^{1,3}$ is $\Ext^2\HV$ and not $\Ext^2V$ or $\Ext^2V^*$, which, $V$ and $V^*$ being totally isotropic, are abelian subalgebras of dimension $n(n-1)/2$ each. In fact, as vector spaces, $\Ext^2\HV=\Ext^2V\oplus(V\wedge V^*)\oplus\Ext^2V^*$, of dimension $n(2n-1)=\dim\mathfrak{so}(\HV)$, and the mixed part $V\wedge V^*$ coincides with the image of $\mathfrak{gl}(V)$ under $j_*$ (both have dimension $n^2$, and $X_\Omega=\tfrac12a^k{}_le_k\wedge\theta^l$) (see the last part of the lemma).} $\SpineH$. For $X\in\Ext^2\HV$, $[X,v]\in\HV$ for $v\in\HV$, and $d\rho:X\mapsto[X,\cdot\,]|_{\HV}$ is a Lie algebra isomorphism onto $\mathfrak{so}(\HV,\langle\,,\,\rangle)$ with inverse
\begin{equation}\label{eq:norm}
\Omega\longmapsto X_\Omega=\tfrac14\,\Omega_{AB}\,E^A\wedge E^B,\qquad \Omega_{AB}:=\langle E_A,\Omega E_B\rangle,
\end{equation}
for any basis $\{E_A\}_{A=1}^{2n}$ of $\HV$ with reciprocal basis $\{E^A\}$ (i.e.\ $\langle E^A,E_B\rangle=\delta^A_B$). For $\Omega=j_*(a)$, $a\in\mathfrak{gl}(V)$, $ae_l=a^k{}_le_k$, one has $X_\Omega=\tfrac12a^k{}_l\,e_k\wedge\theta^l$. Moreover $\tilde X=-X$ for $X\in\Ext^2\HV$, and for every $\Omega\in\mathfrak{so}(\HV)$ the derivation of $\RH$ extending $\Omega$ is $[X_\Omega,\cdot\,]$.
\end{lemma}
\begin{proof}
From \eqref{eq:anticomm}, $[a\wedge b,v]=2(\langle b,v\rangle a-\langle a,v\rangle b)\in\HV$ for $a,b,v\in\HV$, which is skew for $\langle\,,\,\rangle$. A direct substitution gives $[X_\Omega,E_D]=\Omega_{AD}E^A=\Omega E_D$. For $X\in\Ext^2\HV$, $u(t)=\exp(tX)$ is even, $u\,v\,u^{-1}=e^{t\,\mathrm{ad}_X}v\in\HV$ and $\tilde u=\exp(-tX)=u^{-1}$, so $u(t)\in\Gamma^+_{H}$ with $N(u)=1$, a path from $1$: $u(t)\in\SpineH$. Since $\dim\Ext^2\HV=n(2n-1)=\dim\SpineH$ ($\ker\rho$ is discrete), $\Ext^2\HV$ is the Lie algebra. If $[X,\cdot]=0$ on $\HV$, $X$ is central, hence scalar, hence $0$: $d\rho$ is injective. For the Witt basis, $[\tfrac12a^k{}_le_k\wedge\theta^l,e_m]=a^k{}_me_k$ and $[\tfrac12a^k{}_le_k\wedge\theta^l,\theta^m]=-a^m{}_l\theta^l$, which is $j_*(a)=a\oplus(-a^*)$. The last statement holds because $[X_\Omega,\cdot]$ is a derivation of $\RH$ that agrees with $\Omega$ on the generating subspace $\HV$.
\end{proof}

\subsection{Canonical lift of $\mathrm{GL}^+(V)$}

We now present the construction of the canonical lift $\tilde\jmath:\mathrm{GL}^+(V)\to\SpineH$ and establishes the invariance of the spinor pairing under $\SpineH$. 

\begin{proposition}[Canonical lift of $\mathrm{GL}^+(V)$]\label{prop:lift}
There is a unique Lie group homomorphism $\tilde\jmath:\mathrm{GL}^+(V)\to\SpineH$ with $\rho\circ\tilde\jmath=j$. Its differential is $a\mapsto\tfrac12a^k{}_le_k\wedge\theta^l$, and for $A\in\mathrm{GL}^+(V)$, $\alpha\in\Ext V^*$:
\begin{equation}\label{eq:lift}
\begin{gathered}
\tilde\jmath(A)\,\alpha\,\tilde\jmath(A)^{-1}=(A^{-1})^*\alpha,\qquad \tilde\jmath(A)\fH=\fH\tilde\jmath(A)=(\det A)^{1/2}\fH,\\ \tilde\jmath(A)(\alpha\fH)=(\det A)^{1/2}\big((A^{-1})^*\alpha\big)\fH .
\end{gathered}
\end{equation}
\end{proposition}
\begin{proof}
We first recall the domains and ranges of the maps involved. The isotropic extension \eqref{eq:j}, $j:\mathrm{GL}(V)\to O_H$, $j(A)=A\oplus(A^{-1})^*$, is an injective homomorphism (the block $[A]$ is recovered from the block-diagonal matrix of $j(A)$). Since $\mathrm{GL}^+(V)$ is connected, it restricts to an injective homomorphism
\[
j:\mathrm{GL}^+(V)\longrightarrow\SOeH ,
\]
whose image $j(\mathrm{GL}^+(V))\simeq\mathrm{GL}^+(V)$ is a subgroup of dimension $n^2<n(2n-1)=\dim\SOeH$. On the other hand, by \eqref{eq:quotients}, 
\[
\rho:\SpineH\longrightarrow\SOeH,\qquad \rho_u(v)=uvu^{-1},
\]
is a surjective homomorphism with kernel $\{\pm1\}$, i.e.\ a two-sheeted covering. Thus both maps take values in $\SOeH$ and the statement asks for a lift of $j$ through the covering $\rho$:
\[
\begin{array}{ccc}
 & & \SpineH\\[2pt]
 & {\scriptstyle\tilde\jmath}\nearrow & \big\downarrow{\scriptstyle\rho}\\[2pt]
\mathrm{GL}^+(V) & \xrightarrow{\ \ j\ \ } & \SOeH
\end{array}
\]
Since $\mathrm{GL}^+(V)$ is connected, such a lift with $\tilde\jmath(1)=1$ is unique if it exists and it exists if and only if the restricted covering
\[
\rho:\ \rho^{-1}(j(\mathrm{GL}^+(V)))\cap\SpineH\longrightarrow j(\mathrm{GL}^+(V))\simeq\mathrm{GL}^+(V)
\]
is trivial. This restricted map is again a two-sheeted covering (it contains $\ker\rho=\{\pm1\}$), and it is trivial iff the lift of a generator of $\pi_1(\mathrm{GL}^+(V))\simeq\pi_1(SO(n))$ is a closed loop. (That the spin covering is trivial over $j(\mathrm{GL}(V))$ was already observed by Gualtieri \cite[\S2.3, p.~10]{Gualtieri2004}, who also obtains the action $\sqrt{\det g}\,(g^*)^{-1}$ of $\mathrm{GL}^+(V)$ on $\Ext V^*$, cf.\ also the canonical lift $O(n)\to\mathrm{Spin}(n,n)$ of Karoubi \cite[p.~174]{Karoubi1968}. For $n=2$, $\pi_1(\mathrm{GL}^+(V))\simeq\mathbb Z$ is generated by the same rotation and the argument below applies verbatim.) Take a basis orthonormal for an auxiliary inner product and the rotation generator $a=E_{21}-E_{12}$. By Lemma~\ref{lem:lie} and \eqref{eq:witt-from-orth}, $X_a=\tfrac12(e_2\theta^1-e_1\theta^2)=\tfrac12(-\sigma_1\sigma_2+\sigma_{n+1}\sigma_{n+2})$ and the two bivectors commute (and square to $-1$), so
$$\exp(tX_a)=\big(\cos\tfrac t2-\sigma_1\sigma_2\sin\tfrac t2\big)\big(\cos\tfrac t2+\sigma_{n+1}\sigma_{n+2}\sin\tfrac t2\big),$$
which equals $(-1)(-1)=1$ at $t=2\pi$. Hence the covering is trivial and its identity component maps isomorphically onto $\mathrm{GL}^+(V)$. This defines $\tilde\jmath$. The first identity \eqref{eq:lift} holds infinitesimally by Lemma~\ref{lem:lie}, hence on the connected group. By Theorem~\ref{thm:fH}(ii), $\tilde\jmath(A)$ commutes with $\fH$, so $\tilde\jmath(A)\fH=\fH\tilde\jmath(A)\fH\in\mathbb R\fH$: $\tilde\jmath(A)\fH=\chi(A)\fH$ with $\chi$ a positive character. Infinitesimally, $X_a\fH=\tfrac12a^k{}_l(e_k\theta^l-\delta^l_k)\fH=\tfrac12\operatorname{tr}(a)\fH$, because $e_k\theta^l\fH=-\theta^le_k\fH=0$ for $k\neq l$ and $e_k\theta^k\fH=2\fH$. Hence $\chi=(\det)^{1/2}$ and the last identity follows from the first two.
\end{proof}

\begin{proposition}[Spinor pairing]\label{prop:pairing}
For $\psi,\varphi\in S=\RH\fH$ define $\beta(\psi,\varphi)\in\mathbb R$ by $\tilde\psi\varphi=\beta(\psi,\varphi)\,\theta^*$. Then:
\begin{enumerate}[label=(\roman*)]
\item $\beta$ is well defined and $\beta(\alpha\fH,\alpha'\fH)$ is the coefficient of $\theta^*$ in $(\tilde\alpha\wedge\alpha')_{[n]}$ (the Chevalley--Mukai pairing of $\Ext V^*$);
\item $\beta$ is non-degenerate and $\beta(\varphi,\psi)=(-1)^{n(n-1)/2}\beta(\psi,\varphi)$. For $n=4$ it is symmetric of neutral signature $(8,8)$;
\item $\beta(u\psi,u\varphi)=\beta(\psi,\varphi)$ for $u\in\SpineH$, and $\beta(X\psi,\varphi)+\beta(\psi,X\varphi)=0$ for $X\in\Ext^2\HV$. For $u\in\mathrm{Spin}_H$ one only has $\beta(u\psi,u\varphi)=N(u)\,\beta(\psi,\varphi)=\pm\beta(\psi,\varphi)$ (e.g.\ $u=\sigma_1\sigma_{n+1}$ has $N(u)=-1$ and reverses the sign of $\beta$; cf.\ \cite[\S2.4, p.~11]{Gualtieri2004});
\item for $n$ even $S^+\perp S^-$, for $n$ odd $S^\pm$ are isotropic;
\item on $S_c=(\CRH)\fH$, $h(\psi,\varphi)$ defined by $\widetilde{\psi^\natural}\varphi=h(\psi,\varphi)\theta^*$ ($\natural$ = complex conjugation of $\CRH$, fixing $\RH$) is sesquilinear, non-degenerate, $(-1)^{n(n-1)/2}$-Hermitian and $\SpineH$-invariant.
\end{enumerate}
\end{proposition}
\begin{proof}
(i) $\tilde\psi\varphi=\tilde\fH\tilde\alpha\alpha'\fH=\fcH(\tilde\alpha\alpha')\fH=(\tilde\alpha\wedge\alpha')_{[n]}$ by Theorem~\ref{thm:fH}(vii). 
\par\noindent (ii) By Theorem~\ref{thm:fH}, $\{\theta^J\fH\}_{J\subset[n]}$ is a basis of $S$. Since $\tilde\theta^J=(-1)^{|J|(|J|-1)/2}\theta^J$ and $\theta^J\wedge\theta^K=0$ whenever $J\cap K\neq\emptyset$, the degree-$n$ component of
$\tilde\theta^J\wedge\theta^K$ is nonzero if and only if $K=J^c$. Writing $\theta^J\wedge\theta^{J^c}=\varepsilon(J)\,\theta^*$ with $\varepsilon(J)=\pm1$,
part~(i) gives
\[
  \beta(\theta^J\fH,\theta^K\fH)=(-1)^{|J|(|J|-1)/2}\,\varepsilon(J)\,\delta_{K,J^c}.
\]
In the basis $\{\theta^J\fH\}$, each row and each column of the Gram matrix of $\beta$
has exactly one nonzero entry, equal to $\pm1$, located at $K=J^c$ (note that
$J\mapsto J^c$ is a bijection). Hence the Gram matrix is a signed permutation matrix,
so $\det=\pm1$ and $\beta$ is non-degenerate. Reversion is an $\mathbb R$-linear anti-automorphism and $\beta(\psi,\varphi)\in\mathbb R$, so
\[
  \tilde\varphi\psi=\widetilde{\tilde\psi\varphi}
  =\widetilde{\beta(\psi,\varphi)\,\theta^*}=\beta(\psi,\varphi)\,\tilde\theta^*.
\]
Since $\theta^*=\theta^1\cdots\theta^n$ is a product of $n$ pairwise anticommuting elements, $\tilde\theta^*=\theta^n\cdots\theta^1=(-1)^{n(n-1)/2}\theta^*$.
On the other hand, $\tilde\varphi\psi=\beta(\varphi,\psi)\,\theta^*$ by definition.
Comparing coefficients of $\theta^*\neq0$,
\[
  \beta(\varphi,\psi)=(-1)^{n(n-1)/2}\,\beta(\psi,\varphi).
\]

For $n=4$ $(-1)^{n(n-1)/2}=(-1)^6=1$, so $\beta$ is symmetric. Let $S_r=\{\alpha\fH:\alpha\in\Ext^rV^*\}$. By the formula above, $\beta(S_r,S_s)=0$ unless $r+s=4$, so
\[
  S=(S_0\oplus S_4)\;\oplus\;(S_1\oplus S_3)\;\oplus\;S_2
\]
is a $\beta$-orthogonal decomposition. We compute the signature of each summand. For summands $S_r\oplus S_{4-r}$, $r\in\{0,1\}$, since $r+r\neq4$ and $(4-r)+(4-r)\neq4$, both $S_r$ and $S_{4-r}$ are isotropic and $\beta$ pairs them perfectly. Choose bases $\{e_i\}$ of $S_r$ and $\{f_i\}$ of $S_{4-r}$ with $\beta(e_i,f_j)=\delta_{ij}$. By symmetry of $\beta$, the vectors $e_i\pm f_i$ satisfy $\beta(e_i+f_i,e_i+f_i)=2$, $\beta(e_i-f_i,e_i-f_i)=-2$ and are mutually orthogonal. Hence this summand has signature $(d_r,d_r)$ with $d_r=\binom4r$, that is, $(1,1)$ for $r=0$ and $(4,4)$ for $r=1$. Concerning summand $S_2$, for $\alpha\in\Ext^2V^*$ we have $\tilde\alpha=-\alpha$, so $\beta(\alpha\fH,\alpha'\fH)=-\,c$, where $\alpha\wedge\alpha'=c\,\theta^*$. The vectors
  \[
    \theta^{12}+\theta^{34},\quad \theta^{13}-\theta^{24},\quad \theta^{14}+\theta^{23}
  \]
  satisfy $\alpha\wedge\alpha=+2\theta^*$, while
  \[
    \theta^{12}-\theta^{34},\quad \theta^{13}+\theta^{24},\quad \theta^{14}-\theta^{23}
  \]
  satisfy $\alpha\wedge\alpha=-2\theta^*$. These six vectors are mutually
  $\wedge$-orthogonal (two monomials sharing an index have zero wedge product and
  $(\theta^{ij}+\theta^{kl})\wedge(\theta^{ij}-\theta^{kl})=0$). So the wedge pairing on $\Ext^2\mathbb R^4$ has signature $(3,3)$ and the overall sign change leaves $S_2$ with signature $(3,3)$. Adding up, $\beta$ has signature $(1+4+3,\;1+4+3)=(8,8)$.

\par\noindent (iii) Let $u\in\SpineH$. Since $u$ is even, $\bar u=\tilde u$, so $N(u)=u\tilde u=\pm1$ and as $N$ is continuous, $N(1)=1$ and $\SpineH$ is connected, $N(u)=1$, i.e.\ $\tilde u=u^{-1}$. Moreover $uS\subset S$, since $S$ is a left ideal. Hence $\widetilde{(u\psi)}\,(u\varphi)=\tilde\psi\,\tilde uu\,\varphi=\tilde\psi\varphi=\beta(\psi,\varphi)\theta^*$. For $u\in\mathrm{Spin}_H$ the same computation, with $\tilde uu=N(u)=\pm1$, gives $\beta(u\psi,u\varphi)=N(u)\beta(\psi,\varphi)$. Infinitesimally, $\widetilde{(X\psi)}\,\varphi+\tilde\psi\,X\varphi=\tilde\psi(\tilde X+X)\varphi=0$, because $\tilde X=-X$ (Lemma~\ref{lem:lie}).
\par\noindent (iv) By (i), $\beta(\theta^J\fH,\theta^K\fH)\neq0$ only if $\tilde\theta^J\wedge\theta^K$ has a component of degree $n$, i.e.\ $K=J^c$: $\beta$ pairs degree $r$ with degree $n-r$. Recall that $S^+$ consists of even degrees and $S^-$ of odd degrees. If $n$ is even, $r$ and $n-r$ have the same parity, so $\beta$ only pairs $S^+$ with $S^+$ and $S^-$ with $S^-$: $S^+\perp S^-$. If $n$ is odd, the parities are opposite, so $\beta$ only pairs $S^+$ with $S^-$, hence $\beta|_{S^\pm}=0$, i.e.\ $S^\pm$ are isotropic (and $\beta$ puts them in duality).
\par\noindent (v) Let $\beta_{\mathbb C}$ be the $\mathbb C$-bilinear extension of $\beta$ to $S_c=\mathbb C\otimes S$. Since $\natural$ commutes with reversion, $h(\psi,\varphi)=\beta_{\mathbb C}(\psi^\natural,\varphi)$. Hence $h$ is antilinear in $\psi$, linear in $\varphi$, and non-degenerate because $\beta$ is. Since $\beta$ is real, $\overline{\beta_{\mathbb C}(a,b)}=\beta_{\mathbb C}(a^\natural,b^\natural)$ and (ii) gives $h(\varphi,\psi)=(-1)^{n(n-1)/2}\,\overline{h(\psi,\varphi)}$. Finally, $u\in\SpineH$ is real, so $(u\psi)^\natural=u\psi^\natural$ and invariance follows from (iii).
\end{proof}

\section{The hyperbolic Clifford bundle and hyperbolic spin structures}\label{sec:spinstructure}

Throughout, $M$ is an oriented manifold of dimension $n\geq 3$ and $H(M):=TM\oplus T^*M$ is endowed with the non-degenerate symmetric bilinear form \eqref{eq:hypform}. To distinguish them from the elements $x_*,y_*\in V$ and $x^*,y^*\in V^*$, the elements (and sections) of $TM$ will be denoted by $Y_*,W_*$, those of $T^*M$ by $Y^*,W^*$, and those of $\HM$ by $\mathbf Y=Y_*\oplus Y^*$, $\mathbf W=W_*\oplus W^*$. Thus, fibrewise, \eqref{eq:hypform} reads $\langle\mathbf Y,\mathbf W\rangle=Y^*(W_*)+W^*(Y_*)$.

The hyperbolic Clifford bundle over $M$ is defined, in complete analogy with the classical construction of $\Cl(M,g)$ as a union of fibres \cite[Def.\ 4]{MosnaRodrigues2004}, by
$$\Cl(\HM)=\bigsqcup_{x\in M} \Cl(T_xM\oplus T_x^{*}M,\langle\,,\,\rangle),$$
which is precisely the definition given in \cite[\S 4.2]{NotteCuello2026}. This bundle is built over the \emph{hyperbolic frame bundle}
$$F_H(M)=\bigcup_{x\in M} F_xM\oplus F_x^{*}M,$$
introduced in \cite[\S 4.1]{NotteCuello2026} as a principal bundle with structure group $\mathrm{GL}(2n,\mathbb R)$, constructed from the Whitney sum $TM\oplus T^{*}M$. Reducing the structure group to $O_{H}$ (orthonormal frames with respect to $\langle\,,\,\rangle$) and then to $SO_{H}$ and $\SOeH$, yields the principal bundles $P_{O_H}(M)$, $P_{SO_H}(M)$ and $P_{\SOeH}(M)$ \cite[\S 4.1]{NotteCuello2026}. We make the last reduction explicit.

\smallskip\noindent\emph{Witt-induced frames.} For a frame $e=(e_k)$ of $T_xM$ with dual coframe $(\theta^k)$, let $\iota(e):\HV\to H_x(M)$ be the linear map sending the model Witt basis to $e_k$, $\theta^k$ (we use the same letters in the model and on $M$). It is an isometry for $\langle\,,\,\rangle$ \emph{for every} frame $e$, since $\langle\,,\,\rangle$ is built from the duality pairing only. It sends the model orthonormal basis \eqref{eq:orthonormal} to an orthonormal frame of $H_x(M)$, called the \emph{Witt-induced frame}\footnote{Explicitly, the Witt-induced frame of $e$ is $\sigma_k=\tfrac{1}{\sqrt2}(e_k\oplus\theta^k)$, $\sigma_{n+k}=\tfrac{1}{\sqrt2}(-e_k\oplus\theta^k)$, as in \eqref{eq:orthonormal}. Since $\langle e_k,\theta^l\rangle=\theta^l(e_k)=\delta^l_k$ and $\langle e_k,e_l\rangle=\langle\theta^k,\theta^l\rangle=0$ for any frame, these vectors are orthonormal with no metric on $M$. Thus every frame of $T_xM$, orthonormal or not, induces an orthonormal frame of $H_x(M)$. Under $e\mapsto eA$ the coframe changes by $A^{-1}$, so the induced frame changes by $j(A)$.} of $e$. For $A\in\mathrm{GL}(n,\mathbb R)$ one has $\iota(eA)=\iota(e)\circ j(A)$. We take $P_{\SOeH}(M)$ to be the $\SOeH$-bundle generated by the Witt-induced frames of positively oriented frames,
\begin{equation}\label{eq:PSOeH}
P_{\SOeH}(M)=F^+(M)\times_j\SOeH,
\end{equation}
where $F^+(M)$ is the bundle of positively oriented frames of $TM$. Since $j(\mathrm{GL}^+)\subset\SOeH$, this is well defined. An orthonormal frame at $x$ is regarded as an isometry $\phi:\HV\to H_x(M)$ and $\phi_\#:\RH\to\Cl(H_x(M))$ denotes its extension to an algebra isomorphism.

If $g_{\alpha\beta}:U_\alpha\cap U_\beta\to\SOeH$ are the transition functions of $P_{\SOeH}(M)$ with respect to local sections, the transition functions of the hyperbolic Clifford bundle are the algebra automorphisms $h_{\alpha\beta}(x)=(g_{\alpha\beta}(x))_\#$ of $\RH$. When a hyperbolic spin structure (Definition~\ref{def:spinstructure}) is given, with $\SpineH$-valued transition functions $\tilde g_{\alpha\beta}$ such that $\rho(\tilde g_{\alpha\beta})=g_{\alpha\beta}$, they act by conjugation,
\[
h_{\alpha\beta}(x)\,a=\tilde g_{\alpha\beta}(x)\,a\,\tilde g_{\alpha\beta}(x)^{-1},\qquad a\in\RH
\]
(cf.\ Proposition~\ref{prop:ClAd}). In all cases and the hyperbolic Clifford bundle is the associated bundle
\begin{equation}\label{eq:ClHM-assoc}
\Cl(\HM)=P_{\SOeH}(M)\times_{\rho}\RH,
\end{equation}
where $\rho$ acts on $\RH$ through the extension of the defining representation to algebra automorphisms \cite[\S 4.2]{NotteCuello2026}.

\smallskip\noindent\emph{Canonical global sections.} By Theorem~\ref{thm:fH}(ii), the element $\fH$ computed in the Witt-induced frame of any frame $e$ of $T_xM$ is independent of $e$. This defines a global section, still denoted $\fH$, of\footnote{Indeed, in the Witt-induced frame of $e$ one has $\fH=2^{-n}(-1)^{n(n-1)/2}e_*\theta^*$ (Theorem~\ref{thm:fH}(ii)). If $e'=eA$, i.e.\ $e'_k=A^l{}_k e_l$, its dual coframe is $\theta'^k=(A^{-1})^k{}_l\theta^l$. Since $V$ and $V^*$ are totally isotropic, the Clifford products $e_*$, $\theta^*$ are exterior products, hence $e'_*=(\det A)\,e_*$ and $\theta'^*=(\det A)^{-1}\theta^*$. The two factors cancel whatever the sign of $\det A$, so no orientation of $M$ is needed, and $\iota(eA)_\#\fH=\iota(e)_\#\,j(A)_\#\fH=\iota(e)_\#\fH$ for every $A\in\mathrm{GL}(n,\mathbb R)$, and $\fH(x)$ is well defined. In the Witt-induced frames of local frames of $TM$ it has constant components, hence it is smooth and on overlaps $U_\alpha\cap U_\beta$, where $e_\beta=e_\alpha A_{\alpha\beta}$, the local expressions coincide by the computation above, so they glue to a global section. At each point it is a primitive idempotent of $\Cl(H_x(M))$. Unlike $\varpi$, $\fH$ is invariant only under $j(\mathrm{GL}(V))$, not under all of $\SOeH$: it encodes the splitting $\HM=TM\oplus T^*M$.} $\Cl(\HM)$. Likewise, since $\det j(A)=1$, the volume element $\varpi$ of \eqref{eq:volume} defines a global section $\varpi$ of $\Cl(\HM)$, and $\HM$ carries a canonical orientation.\footnote{With effect, the Witt-induced frames of two frames $e$ and $eA$ of $T_xM$ differ by $j(A)$, whose matrix $\mathrm{diag}([A],([A]^{-1})^T)$ has determinant $\det A\,\det A^{-1}=1$. Since an isometry $R$ of $\HV$ acts on the top-degree part by $R_\#(\sigma_1\cdots\sigma_{2n})=(\det R)\,\sigma_1\cdots\sigma_{2n}$, the product $\sigma_1\cdots\sigma_{2n}$ is the same in all Witt-induced frames at $x$, for every $A\in\mathrm{GL}(n,\mathbb R)$ (not only $\mathrm{GL}^+$). Choosing local frames of $TM$, $\varpi$ has constant components in the corresponding Witt-induced frames, so it is a smooth global section of $\bigwedge^{2n}\HM\subset\Cl(\HM)$, nowhere zero since $\varpi^2=1$. Equivalently, all Witt-induced frames at $x$ define the same orientation of $H_x(M)$, the one in which $\varpi$ is positive. This orientation exists even if $M$ is not orientable, reflecting the canonical trivialization $\bigwedge^{2n}(TM\oplus T^*M)\simeq\bigwedge^nTM\otimes\bigwedge^nT^*M\simeq M\times\mathbb R$.} We call $\fH$ and $\varpi$ the \emph{canonical idempotent} and the \emph{volume element} of $\Cl(\HM)$.

We now turn our attention to the hyperbolic spin structures, showing that every oriented manifold carries such an structure. 

\begin{definition}[Hyperbolic spin structure]\label{def:spinstructure}
A \emph{hyperbolic spin structure} on $M$ is a principal bundle
$$\Pi_{\Upsilon}:P_{\SpineH}(M)\to M,$$
with structure group $\SpineH$, together with a map
$$\Upsilon:P_{\SpineH}(M)\to P_{\SOeH}(M),$$
such that:
\begin{enumerate}[label=(\alph*)]
\item $\Pi(\Upsilon(p))=\Pi_{\Upsilon}(p)$ for all $p\in P_{\SpineH}(M)$, where $\Pi$ is the projection of $P_{\SOeH}(M)$;
\item $\Upsilon(pu)=\Upsilon(p)\rho_u$ for all $p\in P_{\SpineH}(M)$, $u\in\SpineH$, with $\rho_u(v)=uvu^{-1}$, $v\in\HV$.
\end{enumerate}
An oriented manifold equipped with a hyperbolic spin structure is called a \emph{hyperbolic spin manifold} and every section of $P_{\SpineH}(M)$ is called a \emph{hyperbolic spin frame field} \cite[\S 4.2]{NotteCuello2026}.
\end{definition}

Definition~\ref{def:spinstructure} is given in \cite{NotteCuello2026}, by analogy with the standard notion of spin structure on a Lorentzian manifold \cite{RodriguesCapelas2016,MosnaRodrigues2004,Geroch1,Geroch2}. In the classical case, $P_{\Spine}(M)$ exists if and only if the second Stiefel--Whitney class of $M$ vanishes. In the hyperbolic case there is no such obstruction: by Karoubi \cite[\S1.1, Prop.~1.1.26, pp.~174--175]{Karoubi1968}, a bundle $E=E^+\oplus E^-$ with structure group $O(p,q)$, $E^\pm$ positive (negative) definite, reduces to $\mathrm{Spin}(p,q)$ if and only if $w_1(E^+)+w_1(E^-)=0$ and $w_2(E^+)+w_2(E^-)=0$. For $\HM=TM\oplus T^*M$ both $E^+$ and $E^-$ project isomorphically onto $TM$, so these classes cancel mod~$2$ in the Whitney sum, and the hyperbolic frame bundle always admits the required lifting \cite[\S2.8, Props.~2.27--2.28, pp.~15--16]{Gualtieri2004} (for the general lifting formalism see also \cite{GP}). The following theorem realizes, in the formalism of $\RH$ and $\Cl(\HM)$, the explicit construction of these spin structures from the frame bundle of $TM$ indicated by Gualtieri \cite[\S2.8, p.~16]{Gualtieri2004} (cf.\ also the canonical lift $O(n)\to\mathrm{Spin}(n,n)$ of Karoubi \cite[p.~174]{Karoubi1968}), and exhibits a canonical hyperbolic spin structure.

\begin{theorem}[Canonical hyperbolic spin structure]\label{thm:canonical}
Let $M$ be oriented and let $F^+(M)$ be its bundle of positively oriented frames. Then
$$\Pcan:=F^+(M)\times_{\tilde\jmath}\SpineH,\qquad [e,u]\sim[eA,\tilde\jmath(A)^{-1}u],$$
$$[e,u]\cdot w:=[e,uw],\qquad \Upsilon([e,u]):=\iota(e)\circ\rho_u,$$
with $\tilde\jmath$ the canonical lift of Proposition~\ref{prop:lift}, is a hyperbolic spin structure in the sense of Definition~\ref{def:spinstructure}. It depends only on the orientation of $M$. For connected $M$, $P_{SO_H}(M)$ has exactly two reductions to $\SOeH$, namely $F^{\pm}(M)\times_j\SOeH$, one for each orientation. Any other hyperbolic spin structure over $P_{\SOeH}(M)$ differs from it by a principal $\mathbb Z_2$-bundle, i.e.\ by a class in $H^1(M;\mathbb Z_2)$.
\end{theorem}
\begin{proof}
$\Upsilon$ is well defined: $\iota(eA)\rho_{\tilde\jmath(A)^{-1}u}=\iota(e)j(A)j(A)^{-1}\rho_u=\iota(e)\rho_u$. Condition (a) of Definition~\ref{def:spinstructure} is clear, and (b): $\Upsilon([e,u]w)=\iota(e)\rho_{uw}=\Upsilon([e,u])\rho_w$. Surjectivity onto $P_{\SOeH}(M)=F^+(M)\times_j\SOeH$ follows from that of $\rho$, and $\Upsilon$ is a two-sheeted covering. Reductions of $P_{SO_H}(M)$ to $\SOeH$ correspond to sections of the principal $\mathbb Z_2$-bundle $P_{SO_H}(M)/\SOeH\to M$, which is trivial because $F^+(M)\times_j\SOeH$ is one of them. On connected $M$ there are thus exactly two, and $F^-(M)\times_j\SOeH$ is the other one, since $j(A)\in SO_H\setminus\SOeH$ when $\det A<0$ (for $A=\mathrm{diag}(-1,1,\dots,1)$, $j(A)$ reverses both $\sigma_1$ and $\sigma_{n+1}$). The last statement is the standard classification of lifts to a double covering group: by the argument of \cite[Theorem~1.7, p.~82]{LawsonMichelsohn} (via Theorem~1.4, Corollary~1.5 and the exact sequence (1.4) there), lifts of $P_{\SOeH}(M)$ correspond to classes in $H^1(P_{\SOeH}(M);\mathbb Z_2)$ restricting on the fibre to the class of the covering $\rho:\SpineH\to\SOeH$, and by exactness of $0\to H^1(M;\mathbb Z_2)\to H^1(P_{\SOeH}(M);\mathbb Z_2)\to H^1(\SOeH;\mathbb Z_2)$ the group $H^1(M;\mathbb Z_2)$ acts freely and transitively on the set of these classes. The argument uses only that $\rho$ is a two-sheeted covering, not the signature.
\end{proof}

\begin{corollary}[Existence of hyperbolic spin structures]\label{cor:existence}
Every oriented manifold $M$ of dimension $n\geq 3$ admits a hyperbolic spin structure, in agreement with \cite[Prop.~1.1.26]{Karoubi1968} and \cite[Prop.~2.28]{Gualtieri2004}. This contrasts with the classical case, where $P_{\Spine}(M)$ exists if and only if $w_2(M)=0$.
\end{corollary}

\begin{remark}[Orientability]\label{rem:orient}
Orientability of $M$ is used only to reduce $P_{SO_H}(M)$ to the identity component $\SOeH$ (condition $w_1(M)=0$). Thus Corollary~\ref{cor:existence} states that there is no topological obstruction beyond orientability. For the full group $\mathrm{Spin}_H$ no condition at all is needed: by Karoubi's criterion \cite[Prop.~1.1.26]{Karoubi1968}, $TM\oplus T^*M$ admits a $\mathrm{Spin}(n,n)$-structure on every manifold, orientable or not \cite[Prop.~2.28]{Gualtieri2004} and the corresponding spinor bundles are $\Ext T^*M\otimes|\det TM|^{1/2}$ and $\Ext T^*M\otimes\det TM\,|\det TM|^{-1/2}$ \cite[\S2.8, p.~16]{Gualtieri2004}. Extending the constructions of this paper to non-orientable manifolds, using $\mathrm{Spin}_H$ and $|\det|^{1/2}$ instead of $\SpineH$ and $(\det)^{1/2}$, is left for future work.
\end{remark}

\begin{remark}
The underlying reason for Theorem~\ref{thm:canonical} is Proposition~\ref{prop:lift}: the double covering $\SpineH\to\SOeH$ is trivial over $j(\mathrm{GL}^+(V))$. This is in accordance with the description of the spinors of $V\oplus V^*$ as $\Ext V^*\otimes(\det V)^{1/2}$ in generalized geometry \cite{Gualtieri2004}; see Theorem~\ref{thm:fock} below.
\end{remark}

\section{Hyperbolic spinor bundles and spin-Clifford bundles}\label{sec:bundles}

In this section $(P_{\SpineH}(M),\Upsilon)$ is a hyperbolic spin structure on $M$ (for instance the canonical one of Theorem~\ref{thm:canonical}).

\begin{definition}[Left real hyperbolic spinor bundle]\label{def:SH}
$$S_H(M)=P_{\SpineH}(M)\times_{\mu_l}\mathcal M_H,$$
where $\mathcal M_H$ is a left module for $\RH$ and $\mu_l$ is the representation of $\SpineH$ on $\mathcal M_H$ given by left multiplication, $\mu_l(u)m=um$.
\end{definition}

\begin{definition}[Right real hyperbolic spinor bundle]\label{def:SHstar}
$$S_H^{\star}(M)=P_{\SpineH}(M)\times_{\mu_r}\mathcal M_H^{\star},$$
where $\mathcal M_H^{\star}$ is a right module for $\RH$ and $\mu_r(u)a=au^{-1}$.
\end{definition}

It is verified directly that $\mu_r$ defines a left representation of the group, despite being given by right multiplication, thanks to the inverse: $\mu_r(uu')a=a(uu')^{-1}=au'^{-1}u^{-1}=\mu_r(u)\mu_r(u')a$. These constructions extend the left and right real spinor bundles defined over $\mathbb R_{1,3}$ by Mosna and Rodrigues \cite[Defs.\ 11--12]{MosnaRodrigues2004}, replacing the classical algebra by $\RH$ and the group $\Spine$ by $\SpineH$. The same construction, for the classical case $\Cl(V,g)$, is also given by Hassan \cite[Defs.\ 9--10]{Hassan2021}.

\begin{definition}[Complex hyperbolic spinor bundles]\label{def:ScH}
In analogy with the classical complex spinor bundles of \cite[Defs.\ 13--14]{MosnaRodrigues2004}, the \emph{complex left} and \emph{complex right hyperbolic spinor bundles} are
$$S_{c,H}(M)=P_{\SpineH}(M)\times_{\mu_c}\mathcal M_{c,H},\qquad S_{c,H}^{\star}(M)=P_{\SpineH}(M)\times_{\mu_c^\star}\mathcal M_{c,H}^{\star},$$
where $\mathcal M_{c,H}$ (resp.\ $\mathcal M_{c,H}^{\star}$) is a complex left (resp.\ right) module for $\CRH$, $\mu_c(u)m=um$ and $\mu_c^\star(u)a=au^{-1}$.
\end{definition}

By Proposition~\ref{prop:realization}, $\CRH\simeq M_{2^n}(\mathbb C)$, so that the simple modules are explicit: one may take for $\mathcal M_{c,H}$ a minimal left ideal of $\CRH$, for instance $(\CRH)\fH\simeq\mathbb C\otimes\Ext V^*\simeq\mathbb C^{2^n}$ (for $n=4$, $\mathbb C^{16}$), the hyperbolic analogue of the column space $\mathbb C^4$ used for $S_c(M)$ in \cite[Def.~13]{MosnaRodrigues2004}, with $\mu_c$ acting by left multiplication of $\SpineH\subset(\RH)^{(0)}$. This choice is studied in Section~\ref{sec:ideal}.

\subsection{Left and right hyperbolic spin-Clifford bundles}

We now construct the hyperbolic extension of the spin-Clifford bundles classically defined by Mosna and Rodrigues \cite[\S\S 2.3--2.4]{MosnaRodrigues2004} (see also \cite[\S\S 6.1--6.3.2]{RodriguesCapelas2016}).

\begin{definition}[Left hyperbolic spin-Clifford bundle]\label{def:Cll}
The left real hyperbolic spin-Clifford bundle of $M$ is the vector bundle
$$\Cl^{l}_{\SpineH}(M)=P_{\SpineH}(M)\times_l \RH,$$
where $l$ is the representation of $\SpineH$ on $\RH$ given by $l(u)x=ux$. Sections of $\Cl^{l}_{\SpineH}(M)$ are called \emph{left hyperbolic spin-Clifford fields}.
\end{definition}

\begin{definition}[Right hyperbolic spin-Clifford bundle]\label{def:Clr}
The right real hyperbolic spin-Clifford bundle of $M$ is
$$\Cl^{r}_{\SpineH}(M)=P_{\SpineH}(M)\times_r \RH,$$
where $r(u)x=xu^{-1}$. Sections are called \emph{right hyperbolic spin-Clifford fields}, in analogy with \cite[Def.\ 20]{MosnaRodrigues2004}.
\end{definition}

The \emph{complexified spin-Clifford bundles} are $\Cl^{l}_{\SpineH}(M)_{\mathbb C}:=P_{\SpineH}(M)\times_l(\CRH)$ and $\Cl^{r}_{\SpineH}(M)_{\mathbb C}:=P_{\SpineH}(M)\times_r(\CRH)$, in analogy with \cite[Rem.\ 28]{MosnaRodrigues2004}. By Proposition~\ref{prop:realization} their typical fibre is $M_{2^n}(\mathbb C)$. Local representatives readily follow from the construction. Given a local section $\zeta:U\to P_{\SpineH}(M)$, sections of $\Cl^l_{\SpineH}(M)$, of $\Cl^r_{\SpineH}(M)$ and of $P_{\SpineH}(M)\times_{\mathrm{Ad}}\RH$ over $U$ are written $[\zeta,\psi_\zeta]$ with $\psi_\zeta:U\to\RH$. For $u:U\to\SpineH$ the definitions give
\begin{equation}\label{eq:transf}
\psi_{\zeta u}=u^{-1}\psi_\zeta\ \ (\Cl^l),\qquad \psi_{\zeta u}=\psi_\zeta u\ \ (\Cl^r),\qquad \psi_{\zeta u}=u^{-1}\psi_\zeta u\ \ (\mathrm{Ad}).
\end{equation}

Moreover, the right and left actions of the algebra are given as follow. For every constant $a\in\CRH$, the map $R_a[p,x]:=[p,xa]$ is well defined on (the complexification of) $\Cl^l_{\SpineH}(M)$, because right multiplication commutes with $l(u)$. This is a free right action of $\RH$ on $\Cl^l_{\SpineH}(M)$ in direct analogy with \cite{MosnaRodrigues2004}. Likewise $L_a[p,x]:=[p,ax]$ is a well-defined left action on $\Cl^r_{\SpineH}(M)$. There are natural inclusions $P_{\SpineH}(M)\hookrightarrow \Cl^{l}_{\SpineH}(M)$, $p\mapsto[p,1]$, and $P_{\SpineH}(M)\hookrightarrow \Cl^{r}_{\SpineH}(M)$, coming from $\SpineH\hookrightarrow(\RH)^{(0)}$. In particular every local section $\zeta$ of $P_{\SpineH}(M)$ defines the local section $\Xi_\zeta:=[\zeta,1]$ of $\Cl^{l}_{\SpineH}(M)$, which is invertible at every point.

\begin{proposition}\label{prop:ClAd}
On a hyperbolic spin manifold, the map $\kappa:P_{\SpineH}(M)\times_{\mathrm{Ad}}\RH\to\Cl(\HM)$, $\kappa([p,a])=\Upsilon(p)_\#(a)$, is a well-defined isomorphism of bundles of algebras. Hence
$$\Cl(\HM)\simeq P_{\SpineH}(M)\times_{\mathrm{Ad}}\RH .$$
\end{proposition}
\begin{proof}
Since $\rho_u(v)=uvu^{-1}$ on $\HV$, its Clifford extension is $(\rho_u)_\#=\mathrm{Ad}_u$. Hence $\kappa([pu,u^{-1}au])=(\Upsilon(p)\rho_u)_\#(u^{-1}au)=\Upsilon(p)_\#(a)$. On each fibre $\kappa$ is an algebra isomorphism. In local trivializations it is smooth with smooth inverse. As in the classical case \cite[Prop.\ 23]{MosnaRodrigues2004}, $\mathrm{Ad}_{-1}=\mathrm{id}$, so $\mathrm{Ad}$ descends to a representation of $\SOeH=\SpineH/\mathbb Z_2$, consistently with \eqref{eq:ClHM-assoc}.
\end{proof}

\subsection{Module bundle structure}

In this subsection we explore a fiber bundle whose fibers carry a module structure over the algebra bundle, so that spinor fields may be described.  

\begin{proposition}\label{prop:module}
$S_H(M)$ (in particular $\Cl^{l}_{\SpineH}(M)$) is a left module bundle over the algebra bundle $\Cl(\HM)$. The same holds for $S_{c,H}(M)$ and $\Cl^{l}_{\SpineH}(M)_{\mathbb C}$ over $\mathbb C\otimes\Cl(\HM)$.
\end{proposition}

\begin{proof}
By Definitions~\ref{def:SH} and~\ref{def:Cll},
$$S_H(M)=P_{\SpineH}(M)\times_{\mu_l}\mathcal M_H,\qquad \Cl^{l}_{\SpineH}(M)=P_{\SpineH}(M)\times_l\RH,$$
are bundles associated with the same principal bundle $P_{\SpineH}(M)$, with $\mu_l(u)m=um$ for $u\in\SpineH\subset\mathbb R^{H\star}$ and $m\in\mathcal M_H$ (respectively $l(u)x=ux$, the particular case $\mathcal M_H=\RH$). Their elements are equivalence classes
$$[p,m]\in S_H(M),\qquad [p,m]\sim[pu,u^{-1}m],\qquad u\in\SpineH$$
and, by Proposition~\ref{prop:ClAd},
$$[p,a]\in\Cl(\HM),\qquad [p,a]\sim[pu,u^{-1}au],\qquad u\in\SpineH.$$

\emph{Fibrewise definition of the action.} Let $x\in M$ and let $[p,a]\in\Cl(\HM)_x$, $[q,m]\in S_H(M)_x$. Since $\SpineH$ acts transitively on the fibre $\Pi_\Upsilon^{-1}(x)$, there is $u\in\SpineH$ with $q=pu$. Rewriting $[q,m]=[pu,m]=[p,um]$, both classes admit a common representative $p$. We define
\begin{equation}\label{eq:action}
[p,a]\cdot[p,m]:=[p,am]\in S_H(M)_x,
\end{equation}
where $am$ denotes the action of $a\in\RH$ on $m\in\mathcal M_H$ given by the module structure (or, for $\Cl^{l}_{\SpineH}(M)$, the Clifford product).

\emph{Well-definedness.} Let $p'=pu$, $u\in\SpineH$, be another common representative of $[p,a]$ and $[p,m]$. Then $[p',a']=[p,a]$ with $a'=u^{-1}au$ and $[p',m']=[p,m]$ with $m'=u^{-1}m$. By associativity of the module action,
$$a'm'=(u^{-1}au)(u^{-1}m)=u^{-1}(am),$$
so that $[p',a'm']=[pu,u^{-1}(am)]=[p,am]$. Thus \eqref{eq:action} is independent of the common representative.

\emph{Module axioms.} For $[p,a],[p,a']\in\Cl(\HM)_x$ and $[p,m],[p,m']\in S_H(M)_x$ (with a common representative $p$, always attainable as above),
\begin{align*}
[p,a]\cdot\big([p,m]+[p,m']\big)&=[p,am]+[p,am']=[p,a]\cdot[p,m]+[p,a]\cdot[p,m'],\\
\big([p,a]+[p,a']\big)\cdot[p,m]&=[p,am]+[p,a'm]=[p,a]\cdot[p,m]+[p,a']\cdot[p,m],\\
\big([p,a][p,a']\big)\cdot[p,m]&=[p,(aa')m]=[p,a(a'm)]=[p,a]\cdot\big([p,a']\cdot[p,m]\big),\\
[p,1]\cdot[p,m]&=[p,m],
\end{align*}
where in the third line we used Proposition~\ref{prop:ClAd} to identify $[p,a][p,a']=[p,aa']$ and the associativity of the action of $\RH$ on $\mathcal M_H$.

\emph{Smoothness.} Let $\{\zeta_\alpha:U_\alpha\to P_{\SpineH}(M)\}$ be an atlas of local sections. Taking $p=\zeta_\alpha(x)$ as common representative for every $x\in U_\alpha$, one obtains smooth local trivializations $\Cl(\HM)|_{U_\alpha}\cong U_\alpha\times\RH$, $S_H(M)|_{U_\alpha}\cong U_\alpha\times\mathcal M_H$, in which \eqref{eq:action} reads $(x,a)\cdot(x,m)=(x,am)$: a fixed bilinear map between finite-dimensional vector spaces. Since the transition functions of $\Cl(\HM)$ and of $S_H(M)$ are given by $\mathrm{Ad}$ and by $\mu_l$ (both smooth), \eqref{eq:action} glues to a smooth map of vector bundles $\Cl(\HM)\otimes S_H(M)\to S_H(M)$. Hence $\sec S_H(M)$ is a left module over the ring $\sec\Cl(\HM)$, via $(\phi\Psi)(x):=\phi(x)\cdot\Psi(x)$. The complex case is identical. This argument is the hyperbolic analogue of the proof of the classical case in \cite[Prop.\ 24]{MosnaRodrigues2004} (cf.\ \cite[p.\ 97]{LawsonMichelsohn}).
\end{proof}

\begin{corollary}\label{cor:psiPhi}
Let $\Phi\in\sec\Cl^l_{\SpineH}(M)$ be invertible at every point (i.e.\ $\Phi_\zeta(x)\in\mathbb R^{H\star}$, a condition independent of $\zeta$). Then for every $\Psi\in\sec\Cl^l_{\SpineH}(M)$ there is a unique $\psi\in\sec\Cl(\HM)$ with $\Psi=\psi\Phi$, namely $\psi=[\zeta,\Psi_\zeta\Phi_\zeta^{-1}]_{\mathrm{Ad}}$. In particular, locally every $\Psi$ can be written as $\Psi=\psi\,\Xi_\zeta$ with $\psi\in\sec\Cl(H(U))$. The invertibility hypothesis cannot be dropped: if $\Phi$ takes values in an ideal subbundle ($\Phi=\Phi\fH$, see Section~\ref{sec:ideal}), then so does $\psi\Phi$, and no $\Psi$ outside that subbundle is of the form $\psi\Phi$.
\end{corollary}
\begin{proof}
By \eqref{eq:transf}, $\Psi_{\zeta u}\Phi_{\zeta u}^{-1}=u^{-1}\Psi_\zeta\Phi_\zeta^{-1}u$, which is the $\mathrm{Ad}$-rule and $[\zeta,\psi_\zeta]_{\mathrm{Ad}}\cdot[\zeta,\Phi_\zeta]=[\zeta,\psi_\zeta\Phi_\zeta]$ by Proposition~\ref{prop:module}. Uniqueness follows from invertibility.
\end{proof}

This is the hyperbolic counterpart of \cite[Cor.\ 25]{MosnaRodrigues2004}. The local representation $\Psi=\psi\,\Xi_\zeta$ is the one underlying the Dirac--Hestenes description of spinor fields.

\begin{definition}\label{def:equiv-spinor}
Two left real hyperbolic spinor bundles are \emph{equivalent} if they are isomorphic as module bundles over $\Cl(\HM)$, in analogy with \cite[Def.\ 26]{MosnaRodrigues2004}.
\end{definition}

\subsection{Ideal hyperbolic algebraic spinor fields}\label{sec:ideal}

\begin{definition}[Real and complex ideal hyperbolic algebraic spinor fields]\label{def:ideal}
For a primitive idempotent $f$ of $\RH$ (resp.\ $\CRH$) put
$$I_f(M):=\{\Psi\in\Cl^l_{\SpineH}(M):\Psi f=\Psi\},\qquad I^\star_f(M):=\{\Psi\in\Cl^r_{\SpineH}(M):f\Psi=\Psi\}$$
(resp.\ inside the complexified bundles). $I_f(M)$ is called a subbundle of \emph{left-ideal hyperbolic algebraic spinor fields} (H-LIASF) and $I^\star_f(M)$ a subbundle of \emph{right-ideal hyperbolic algebraic spinor fields} (H-RIASF), in analogy with the LIASF and RIASF of \cite[Defs.\ 18, 21]{MosnaRodrigues2004}. The \emph{canonical real} and \emph{complex H-LIASF bundles} are
$$\IH:=I_{\fH}(M),\qquad \IcH:=I^{\mathbb C}_{\fH}(M)\subset\Cl^l_{\SpineH}(M)_{\mathbb C},$$
and the canonical H-RIASF bundle is $I^\star_H(M):=I^\star_{\fH}(M)$.
\end{definition}

\noindent In the condition $\Psi f=\Psi$, $f$ is a constant element of $\RH$ acting on $\Cl^l_{\SpineH}(M)$ by right multiplication $R_f$ (see the paragraph on right and left actions of the algebra in Section~\ref{sec:bundles}). This must not be confused with the left action of the global section $\fH$ of $\Cl(\HM)$ (Section~\ref{sec:spinstructure}) through the module structure of Proposition~\ref{prop:module}: both objects are denoted by $\fH$, but the two operations are different.

The condition $\Psi f=\Psi$ is $R_f\Psi=\Psi$, i.e.\ it places $\Psi$ fibrewise in the left minimal ideal generated by $f$. For $f=\fH$, the fibre of $\IH$ is $\RH\fH=\RH\theta^*\simeq\Ext V^*$ (Theorem~\ref{thm:fH}(v)) and that of $I^\star_H(M)$ is $\fH\RH=e_*\RH\simeq\Ext V$ (Theorem~\ref{thm:fH}(vi)).

\begin{definition}\label{def:geomequiv}
Two H-LIASF subbundles $I_f(M)$, $I_{f'}(M)$ (resp.\ H-RIASF subbundles) are \emph{geometrically equivalent} if $f'=ufu^{-1}$ for some $u\in\SpineH$.
\end{definition}

\begin{theorem}[Structure of the ideal subbundles]\label{thm:ideal}
\begin{enumerate}[label=(\roman*)]
\item $I_f(M)=P_{\SpineH}(M)\times_l(\RH f)$ is a subbundle of real rank $2^n$ (complex rank $2^n$ in the complex case), the image of the projection $R_f$. It coincides fibrewise with $\Ext V^*\simeq\RH\theta^*$ for $f=\fH$.
\item $I_f(M)$ is a module bundle over $\Cl(\HM)$ (resp.\ $\mathbb C\otimes\Cl(\HM)$), fibrewise simple. It is the spinor bundle $S_H(M)$ of Definition~\ref{def:SH} with $\mathcal M_H=\RH f$ and, in the complex case, the bundle $S_{c,H}(M)$ of Definition~\ref{def:ScH} with $\mathcal M_{c,H}=(\CRH)f\simeq\mathbb C^{2^n}$ and $\mu_c$ the restriction of $l$.
\item If $f'=gfg^{-1}$, then $R_{g^{-1}}:I_f(M)\to I_{f'}(M)$ is an isomorphism of module bundles. Geometric equivalence (Definition~\ref{def:geomequiv}) is the case $g\in\SpineH$. All real (resp.\ complex) ideal subbundles are isomorphic, and every $S_H(M)$ with simple $\mathcal M_H$ is equivalent, in the sense of Definition~\ref{def:equiv-spinor}, to $\IH$. For arbitrary $\mathcal M_H\simeq(\RH\fH)^{\oplus m}$, $S_H(M)\simeq\IH^{\oplus m}$ (same in the complex case).
\item $\Cl^l_{\SpineH}(M)=\bigoplus_JI_{p_J}(M)$ (Proposition~\ref{prop:idempotents}(b)), each summand isomorphic to $\IH$.
\item $\IcH=\mathbb C\otimes\IH$; complex conjugation $\Psi\mapsto\Psi^\natural$ is a well-defined antilinear involution of $\IcH$ with fixed set $\IH$.
\item $\IH=I^+_H(M)\oplus I^-_H(M)$, with $I^\pm_H(M)=P_{\SpineH}(M)\times_l S^\pm$ (Proposition~\ref{prop:chir}).
\item $\beta([p,\psi],[p,\varphi]):=\beta(\psi,\varphi)$ and $h$ (Proposition~\ref{prop:pairing}) define non-degenerate fibre forms on $\IH$ and $\IcH$, symmetric (resp.\ Hermitian) if $n\equiv0,1\pmod 4$ and skew-symmetric (resp.\ skew-Hermitian) if $n\equiv2,3\pmod4$.
\item $I^\star_f(M)=P_{\SpineH}(M)\times_r(f\RH)$ is the right spinor bundle $S^\star_H(M)$ of Definition~\ref{def:SHstar} with $\mathcal M^\star_H=f\RH$ (for $f=\fH$, fibre $e_*\RH\simeq\Ext V$), and $[p,a]\otimes[p,b]\mapsto ab\in f\RH f=\mathbb Rf$ is a well-defined non-degenerate pairing $I^\star_f(M)\otimes I_f(M)\to M\times\mathbb R$.
\end{enumerate}
\end{theorem}
\begin{proof}
(i) $\RH f$ is invariant under left multiplication, so $P\times_l(\RH f)\subset\Cl^l$ is a subbundle; $\Psi f=\Psi$ iff $\Psi_\zeta\in\RH f$, a condition preserved by \eqref{eq:transf}. The rank is that of a minimal left ideal of $M_{2^n}(\mathbb R)$ and the last statement is Theorem~\ref{thm:fH}(v). (ii) The module structure is that of Proposition~\ref{prop:module} restricted to $\RH f$; simplicity is fibrewise minimality. (iii) $R_{g^{-1}}$ commutes with $l(u)$ and with the left module action and $(\Psi g^{-1})f'=\Psi fg^{-1}=\Psi g^{-1}$. Every simple module of $M_N(\mathbb R)$ is isomorphic to $\RH\fH$ by a constant $\RH$-linear map $T$ and $[p,m]\mapsto[p,Tm]$ is well defined because $T$ commutes with $\mu_l(u)$. (iv) From $\sum_Jp_J=1$ and orthogonality, $\Psi=\sum_J\Psi p_J$ uniquely. Now apply (iii). (v) $\fH$ is real. (vi) $\SpineH$ preserves $S^\pm$. (vii) $\beta$ and $h$ are $\SpineH$-invariant. (viii) $r(u)$ preserves $f\RH$; $(au^{-1})(ub)=ab$; non-degeneracy follows from Theorem~\ref{thm:fH}(v)--(vii).
\end{proof}

\begin{theorem}[Fock realization for the canonical spin structure]\label{thm:fock}
Let $P=\Pcan$ and let $s:\Ext V^*\to\Ext V^*$ be $s(\alpha)=2^{-r/2}\alpha$ for $\alpha\in\Ext^rV^*$ (note that $s$ commutes with every degree-preserving operator on $\Ext V^*$, such as $(A^{-1})^*$ or $\theta^l\wedge i_{e_k}$. This is used below and in Proposition~\ref{prop:explicit}). The map
$$\Phi:\Ext T^*M\otimes(\Ext^nTM)^{1/2}\longrightarrow\IH,\qquad \Phi\big(\alpha\otimes(e_*)^{1/2}\big):=\big[[e,1],\,s(e^*\alpha)\,\fH\big],$$
is a well-defined isomorphism of vector bundles (here $e^*\alpha\in\Ext V^*$ are the components of $\alpha$ in the frame $e$, and $(\Ext^nTM)^{1/2}=F^+(M)\times_{(\det)^{1/2}}\mathbb R$). For $Y^*\in T^*M$ and $Y_*\in TM$, regarded as sections of $\HM$ acting by the module structure,
$$Y^*\cdot\Phi(\omega)=\sqrt2\,\Phi(Y^*\wedge\omega),\qquad Y_*\cdot\Phi(\omega)=\sqrt2\,\Phi(i_{Y_*}\omega),\qquad \varpi\cdot\Phi(\omega)=(-1)^{n(n-1)/2}\Phi(\hat\omega),$$
and $\beta(\Phi(\alpha\otimes\mu),\Phi(\alpha'\otimes\mu'))=2^{-n/2}\big\langle(\tilde\alpha\wedge\alpha')_{[n]},\mu\mu'\big\rangle$. If $\gring$ is a metric on $M$ (any signature), $(e_*)^{1/2}$ for positively oriented $\gring$-orthonormal frames is a global trivialization of $(\Ext^nTM)^{1/2}$, and $\Phi$ gives canonical isomorphisms $\IH\simeq\Ext T^*M\simeq\Cl(M,\gring)$ and $\IcH\simeq\mathbb C\otimes\Ext T^*M$ (as vector bundles).
\end{theorem}
\begin{proof}
In $\Pcan$, $[eA,1]=[e,\tilde\jmath(A)]=[e,1]\tilde\jmath(A)$, so $[[eA,1],x]=[[e,1],\tilde\jmath(A)x]$. With $(eA)^*\alpha=A^*e^*\alpha$, $(eA)_*=(\det A)e_*$ and \eqref{eq:lift} ($s$ commutes with $(A^{-1})^*$),
$$\Phi\big(\alpha\otimes(eA)_*^{1/2}\big)=\big[[e,1],(\det A)^{1/2}s(e^*\alpha)\fH\big]=(\det A)^{1/2}\,\Phi\big(\alpha\otimes e_*^{1/2}\big),$$
which is consistent with linearity; hence $\Phi$ is well defined and it is a fibrewise isomorphism by Theorem~\ref{thm:fH}(v). The Clifford formulas follow from \eqref{eq:fock} and the choice of $s$ (which symmetrizes the factors $1$ and $2$ into $\sqrt2$), the chirality formula from Proposition~\ref{prop:chir}, the pairing from Proposition~\ref{prop:pairing}(i), since $s$ contributes $2^{-r/2}2^{-(n-r)/2}$. For orthonormal frames $\det A=1$, whence the global trivialization.
\end{proof}

\section{Connections}\label{sec:connections}

\subsection{Isotropic extension of a linear connection}

Let $\nabla$ be an arbitrary linear connection on $TM$ (not necessarily compatible with a metric, nor torsion-free). It induces a connection on $T^*M$, still denoted $\nabla$, characterized by the compatibility with the duality pairing: for\footnote{Here $Z$ acts as a vector field differentiating the real-valued function $Y^*(Y_*):M\to\mathbb R$, whose value at $q\in M$ is $Y^*|_q(Y_*|_q)$, i.e.\ $Z(f)=df(Z)$ for $f\in C^\infty(M)$.} $Z,Y_*\in\sec TM$ and $Y^*\in\sec T^*M$,
\begin{equation}\label{eq:dualpairing}
Z\big(Y^*(Y_*)\big)=(\nabla_{Z}Y^*)(Y_*)+Y^*(\nabla_{Z}Y_*).
\end{equation}
This identity holds for every linear connection, with no metric hypothesis. It allows one to extend $\nabla$ diagonally to $\HM=TM\oplus T^*M$.

\begin{definition}[Isotropic extension]\label{def:isotropic}
The \emph{isotropic extension} of $\nabla$ is the connection $\nabla^H$ on $\HM$ given by
\begin{equation}\label{eq:DH}
\nabla^H_{Z}(Y_*\oplus Y^*):=(\nabla_{Z}Y_*)\oplus(\nabla_{Z}Y^*),\qquad Z,Y_*\in\sec TM,\ Y^*\in\sec T^*M .
\end{equation}
When $\nabla=\Dring$ is the Levi-Civita connection of a metric $\gring$ on $M$, we write $\Dring^H$.
\end{definition}

The name is chosen in analogy with the isotropic extension $I(\phi)=\phi\oplus\phi^*$ of an endomorphism \cite[\S 3.1]{RodriguesSouza2007}. If $\nabla_{Z}e_l=\omega_l{}^k(Z)e_k$ in a local frame $e$ of $TM$, then $\nabla_{Z}\theta^k=-\omega_l{}^k(Z)\theta^l$ by \eqref{eq:dualpairing}, so the connection matrix of $\nabla^H$ in the Witt-induced frame of $e$ is $j_*(\omega(Z))=\omega(Z)\oplus(-\omega(Z)^*)$.

\begin{lemma}\label{lem:DHcompat}
For every linear connection $\nabla$ on $TM$, $\nabla^H$ is compatible with $\langle\,,\,\rangle$ and extends uniquely to a connection on $\Cl(\HM)$ satisfying the Leibniz rule for the Clifford product: for $\mathbf Y,\mathbf W\in\sec\HM$, $Z\in\sec TM$ and $A,B\in\sec\Cl(\HM)$,
\[
Z\langle\mathbf Y,\mathbf W\rangle=\langle\nabla^H_{Z}\mathbf Y,\mathbf W\rangle+\langle\mathbf Y,\nabla^H_{Z}\mathbf W\rangle,\qquad
\nabla^H_{Z}(AB)=(\nabla^H_{Z}A)B+A(\nabla^H_{Z}B).
\]
\end{lemma}
\begin{proof}
\noindent\textbf{(i) Compatibility with $\langle\,,\,\rangle$.}
Write $\mathbf Y=Y_*\oplus Y^*$, $\mathbf W=W_*\oplus W^*$, with $Y_*,W_*\in\sec TM$ and $Y^*,W^*\in\sec T^*M$. By \eqref{eq:hypform}, $\langle\mathbf Y,\mathbf W\rangle=Y^*(W_*)+W^*(Y_*)$, so that
\begin{equation}\label{eq:step1}
Z\langle\mathbf Y,\mathbf W\rangle=Z\big(Y^*(W_*)\big)+Z\big(W^*(Y_*)\big).
\end{equation}
By \eqref{eq:dualpairing},
$Z\big(Y^*(W_*)\big)=(\nabla_{Z}Y^*)(W_*)+Y^*(\nabla_{Z}W_*)$ and $Z\big(W^*(Y_*)\big)=(\nabla_{Z}W^*)(Y_*)+W^*(\nabla_{Z}Y_*)$.
On the other hand, by \eqref{eq:DH}, $\nabla^H_{Z}\mathbf Y=(\nabla_{Z}Y_*)\oplus(\nabla_{Z}Y^*)$ and $\nabla^H_{Z}\mathbf W=(\nabla_{Z}W_*)\oplus(\nabla_{Z}W^*)$, so that
\[
\langle\nabla^H_{Z}\mathbf Y,\mathbf W\rangle=(\nabla_{Z}Y^*)(W_*)+W^*(\nabla_{Z}Y_*),\qquad
\langle\mathbf Y,\nabla^H_{Z}\mathbf W\rangle=Y^*(\nabla_{Z}W_*)+(\nabla_{Z}W^*)(Y_*).
\]
The sum of these two expressions is exactly the right-hand side of \eqref{eq:step1}, rewritten with \eqref{eq:dualpairing}. This proves the first identity.

\smallskip
\noindent\textbf{(ii) Leibniz rule for the Clifford product.}
We first extend $\nabla^H$ to the tensor algebra $T(\HM)=\bigoplus_{k\geq0}\HM^{\otimes k}$ and then check that the extension descends to the quotient defining the Clifford bundle (cf.\ Eq.~(4.150) of \cite{RodriguesCapelas2016}). For sections $\mathbf Y_1,\dots,\mathbf Y_k$ of $\HM$ define
\[
\nabla^H_{Z}(\mathbf Y_1\otimes\cdots\otimes\mathbf Y_k):=\sum_{i=1}^k\mathbf Y_1\otimes\cdots\otimes(\nabla^H_{Z}\mathbf Y_i)\otimes\cdots\otimes\mathbf Y_k,
\]
extended by linearity, with $\nabla^H_{Z} f=Z(f)$ on functions. This extension satisfies
\begin{equation}\label{eq:leibniz-tensor}
\nabla^H_{Z}(S\otimes T)=(\nabla^H_{Z}S)\otimes T+S\otimes(\nabla^H_{Z}T),\qquad S,T\in\sec T(\HM).
\end{equation}

Now, consider $\Cl(\HM)=T(\HM)/\mathcal I$, where $\mathcal I$ is the bundle of two-sided ideals generated fibrewise by $\mathbf Y\otimes \mathbf W+\mathbf W\otimes \mathbf Y-2\langle \mathbf Y,\mathbf W\rangle\,\mathbf 1$. Using \eqref{eq:leibniz-tensor}, part (i) and $\nabla^H_{Z}\mathbf1=0$,
\begin{align*}
\nabla^H_{Z}\big(\mathbf Y\otimes \mathbf W+\mathbf W\otimes \mathbf Y-2\langle \mathbf Y,\mathbf W\rangle\,\mathbf1\big)
&=\Big[(\nabla^H_{Z}\mathbf Y)\otimes \mathbf W+\mathbf W\otimes(\nabla^H_{Z}\mathbf Y)-2\langle\nabla^H_{Z}\mathbf Y,\mathbf W\rangle\,\mathbf1\Big]\\
&\qquad+\Big[\mathbf Y\otimes(\nabla^H_{Z}\mathbf W)+(\nabla^H_{Z}\mathbf W)\otimes \mathbf Y-2\langle \mathbf Y,\nabla^H_{Z}\mathbf W\rangle\,\mathbf1\Big]
\end{align*}
and each bracket is a generator of $\mathcal I$. Since $\mathcal I$ is a two-sided ideal, \eqref{eq:leibniz-tensor} shows that $\nabla^H_{Z}(\mathcal I)\subset\mathcal I$. Hence $\nabla^H$ induces a connection on the quotient and passing \eqref{eq:leibniz-tensor} to the quotient (where $\otimes$ becomes the Clifford product) gives $\nabla^H_{Z}(AB)=(\nabla^H_{Z}A)B+A(\nabla^H_{Z}B)$. Uniqueness holds because $\Cl(\HM)$ is locally generated by $\HM$.
\end{proof}

\begin{remark}[Metric compatibility and non-metricity]\label{rem:relocation}
Lemma~\ref{lem:DHcompat} shows that even when $\nabla$ is not compatible with a given metric $\gring$ on $M$ ($\nabla\gring\neq0$, as in Weyl or Riemann--Cartan--Weyl geometries), its isotropic extension $\nabla^H$ is \emph{always} compatible with the hyperbolic metric of $\HM$. The non-metricity of $\nabla$ is thus not detected by $\langle\,,\,\rangle$, but by the splitting of $\HM$ determined by $\gring$: put $H_\pm:=\{Y_*\oplus\pm\gring^\flat Y_*:Y_*\in TM\}$, so that $\HM=H_+\oplus H_-$ orthogonally and $\langle\,,\,\rangle|_{H_\pm}=\pm2\gring$. Then $H_+$ and $H_-$ are $\nabla^H$-parallel if and only if $\nabla\gring=0$. Indeed, $\nabla^H_{Z}(Y_*\oplus\pm\gring^\flat Y_*)=\nabla_{Z}Y_*\oplus\pm\nabla_{Z}(\gring^\flat Y_*)$ lies in $H_\pm$ for all $Y_*$ iff $\nabla_{Z}(\gring^\flat Y_*)=\gring^\flat(\nabla_{Z}Y_*)$, i.e.\ iff $(\nabla_{Z}\gring)(Y_*,\cdot)=0$. When $\nabla\gring=0$, $Y_*\mapsto\tfrac1{\sqrt2}(Y_*\oplus\pm\gring^\flat Y_*)$ induces isomorphisms $\Cl(TM,\pm\gring)\simeq\Cl(H_\pm)\subset\Cl(\HM)$ compatible with the connections, the geometric form of \eqref{eq:tensor-iso}. In particular, for the Levi-Civita connection both $H_+$ and $H_-$ are $\Dring^H$-parallel and on each factor $\Cl(H_\pm)\simeq\Cl(TM,\pm\gring)$ of \eqref{eq:tensor-iso} $\Dring^H$ restricts to the connection induced by the Levi-Civita connection.
\end{remark}

\begin{proposition}[Parallel sections]\label{prop:parallel}
Let $\nabla$ be a linear connection on $M$ and $\nabla^H$ its isotropic extension.
\begin{enumerate}[label=(\alph*)]
\item For any local frame $e$ of $TM$, $\nabla_{Z}e_*=\operatorname{tr}\omega(Z)\,e_*$ and $\nabla_{Z}\theta^*=-\operatorname{tr}\omega(Z)\,\theta^*$, and $\nabla^H$ restricted to $\Ext TM$ and to $\Ext T^*M$ (subalgebras of $\Cl(\HM)$) is the connection induced by $\nabla$ on multivectors and on forms.
\item The canonical idempotent $\fH$ is $\nabla^H$-parallel.
\item The volume element $\varpi$ of $\Cl(\HM)$ is $\nabla^H$-parallel.
\item If $\nabla$ is compatible with a metric $\gring$ and $e$ is $\gring$-orthonormal, then $\operatorname{tr}\omega=0$, and $e_*$ and $\theta^*$ are $\nabla$-parallel; in particular this holds for $\Dring$.
\end{enumerate}
\end{proposition}
\begin{proof}
(a) Since $TM$ and $T^*M$ are totally isotropic, the Clifford product on each of them is the exterior product (cf.\ \eqref{eq:witt}), so, by Lemma~\ref{lem:DHcompat}, $\nabla^H$ acts on $\Ext TM$ and $\Ext T^*M$ as the derivation extending $\nabla$. Expanding $\nabla_{Z}e_k=\omega_k{}^l(Z)e_l$ inside $e_1\wedge\cdots\wedge e_n$, only the terms $l=k$ survive, giving $\operatorname{tr}\omega(Z)e_*$, and similarly for $\theta^*$, with $\nabla_{Z}\theta^k=-\omega_l{}^k(Z)\theta^l$.
(b) By Theorem~\ref{thm:fH}(ii), $\fH=2^{-n}e_*\tilde\theta^*$ in the Witt-induced frame of any local frame $e$. By (a) and the Leibniz rule, $\nabla^H_{Z}\fH=2^{-n}\big(\operatorname{tr}\omega(Z)-\operatorname{tr}\omega(Z)\big)e_*\tilde\theta^*=0$.
(c) In a local orthonormal frame $\{\underline\sigma_A\}$ of $\HM$ (e.g.\ a Witt-induced one), $\nabla^H_{Z}\underline\sigma_A=\Omega(Z)^B{}_A\underline\sigma_B$ with $\Omega(Z)\in\mathfrak{so}(\HV)$ by the metric compatibility in Lemma~\ref{lem:DHcompat}. By Lemma~\ref{lem:lie}, on sections with constant components in this frame, $\nabla^H_{Z}$ acts as the derivation $[X_{\Omega(Z)},\cdot\,]$; $\varpi$ has constant components, and it commutes with the even element $X_{\Omega(Z)}$ (Proposition~\ref{prop:chir}). Hence $\nabla^H_{Z}\varpi=0$.
(d) Metric compatibility gives $0=Z\,\gring(e_k,e_l)=\omega_{kl}(Z)+\omega_{lk}(Z)$ with $\omega_{kl}:=\gring_{lm}\omega_k{}^m$, so $\omega_k{}^k(Z)=\gring^{kk}\omega_{kk}(Z)=0$ for each $k$ and $\operatorname{tr}\omega=0$. Now apply (a).
\end{proof}

\begin{remark}
In the Clifford algebra, $e_*\theta^*$ is \emph{not} the exterior product $e_*\wedge\theta^*$ (the factors are not orthogonal, $\langle e_k,\theta^k\rangle=1$). By Theorem~\ref{thm:fH}(ii), $e_*\theta^*=2^n(-1)^{n(n-1)/2}\fH$, which is idempotent up to a factor, whereas the volume element of $\HV$ is $\varpi$, with $\varpi^2=1$. Proposition~\ref{prop:parallel} shows that both are parallel, the former for every linear connection on $M$.
\end{remark}

\subsection{The hyperbolic spin connection}

 Let $(P_{\SpineH}(M),\Upsilon)$ be a hyperbolic spin structure. In this subsection vector fields on $M$ are denoted by $Z,Z_1,Z_2$ to avoid confusion with the model space $V$, with elements $u$ of $\SpineH$ and with open sets $U\subset M$. Let $\{E_A\}_{A=1}^{2n}$ be a basis of $\HV$, so that $E_AE_B+E_BE_A=2\langle E_A,E_B\rangle$ and let $\{E^A\}$ be its reciprocal basis.

\begin{definition}[Hyperbolic spin connection]\label{def:spinconn}
Let $\nabla^H$ be a linear connection on $\HM$ compatible with\footnote{Examples are $\Dring^H$ and the isotropic extension $\nabla^H$ of an arbitrary linear connection $\nabla$ on $M$ (Lemma~\ref{lem:DHcompat}).} $\langle\,,\,\rangle$. Let $\zeta:U\to P_{\SpineH}(M)$ be a local section, $\underline E_A:=\Upsilon(\zeta)(E_A)$ the corresponding orthonormal frame of $\HM$, and $\nabla^H_{Z}\underline E_A=\Omega^\zeta(Z)^B{}_A\underline E_B$. The \emph{local spin connection form} is
$$\omega^\zeta_{Z}:=2X_{\Omega^\zeta(Z)}=\tfrac12\,\Omega^\zeta_{AB}(Z)\,E^A\wedge E^B\in\Ext^2\HV,\qquad\text{i.e.}\qquad \tfrac12[\omega^\zeta_{Z},v]=\Omega^\zeta(Z)v\ \ (v\in\HV),$$
and the \emph{hyperbolic spinorial covariant derivative} on $\Cl^l_{\SpineH}(M)$ (and on its complexification, and on $S_H(M)$, $S_{c,H}(M)$ through $\mu_l$, $\mu_c$) is
\begin{equation}\label{eq:nablas}
\nabla^s_{Z}[\zeta,\psi_\zeta]:=\big[\zeta,\,Z(\psi_\zeta)+\tfrac12\,\omega^\zeta_{Z}\,\psi_\zeta\big],\qquad Z\in\sec TM.
\end{equation}
On $\Cl^r_{\SpineH}(M)$ and $S^\star_H(M)$ we set $\nabla^r_{Z}[\zeta,\varphi_\zeta]:=[\zeta,Z(\varphi_\zeta)-\tfrac12\varphi_\zeta\omega^\zeta_{Z}]$, and on $P_{\SpineH}(M)\times_{\mathrm{Ad}}\RH$, $\nabla_{Z}[\zeta,a_\zeta]:=[\zeta,Z(a_\zeta)+\tfrac12[\omega^\zeta_{Z},a_\zeta]]$.
\end{definition}

Formula \eqref{eq:nablas} is the hyperbolic analogue of the classical spinorial covariant derivative $\nabla^s_{Z}\Psi=Z(\Psi)+\tfrac12\omega_{Z}\Psi$ of \cite[Prop.\ 58]{MosnaRodrigues2004}. The factor $\tfrac12$ corresponds to the normalization $\tfrac12[\omega_{Z},\cdot\,]=\Omega(Z)$. Note that $\omega^\zeta$ is a \emph{local} ($\zeta$-dependent) $\Ext^2\HV$-valued $1$-form; the global object is the principal connection of Theorem~\ref{thm:welldef}(a).

\begin{theorem}[Well-definedness]\label{thm:welldef}
\begin{enumerate}[label=(\alph*)]
\item The forms $\tfrac12\omega^\zeta$ are the local representatives of a unique principal connection on $P_{\SpineH}(M)$ lifting the connection of $P_{\SOeH}(M)$ defined by $\nabla^H$. For $u:U\to\SpineH$,
\begin{equation}\label{eq:gauge}
\omega^{\zeta u}=u^{-1}\omega^\zeta u+2\,u^{-1}du .
\end{equation}
\item The hyperbolic spinorial covariant derivative \eqref{eq:nablas} is independent of $\zeta$ and defines a linear connection on $\Cl^l_{\SpineH}(M)$, $\Cl^l_{\SpineH}(M)_{\mathbb C}$, $S_H(M)$ and $S_{c,H}(M)$. Likewise $\nabla^r$ and the $\mathrm{Ad}$-connection.
\end{enumerate}
\end{theorem}
\begin{proof}
(a) The frame $\Upsilon(\zeta u)=\Upsilon(\zeta)\rho_u$ has connection matrix $\rho_u^{-1}\Omega^\zeta\rho_u+\rho_u^{-1}d\rho_u$. Since $\rho$ is a homomorphism, $\rho_u^{-1}d\rho_u=d\rho(u^{-1}du)$ and $\rho_u^{-1}d\rho(X)\rho_u=d\rho(u^{-1}Xu)$. Applying $d\rho^{-1}$ (Lemma~\ref{lem:lie}) gives $\tfrac12\omega^{\zeta u}=u^{-1}\tfrac12\omega^\zeta u+u^{-1}du$, the transformation law of a connection on the principal bundle. Uniqueness holds because $d\rho$ is a Lie algebra isomorphism.
(b) With $\psi_{\zeta u}=u^{-1}\psi_\zeta$,
$$Z(u^{-1}\psi_\zeta)+\tfrac12(u^{-1}\omega^\zeta_{Z}u+2u^{-1}Z(u))u^{-1}\psi_\zeta=u^{-1}\big(Z(\psi_\zeta)+\tfrac12\omega^\zeta_{Z}\psi_\zeta\big)+\big(Z(u^{-1})+u^{-1}Z(u)u^{-1}\big)\psi_\zeta,$$
and the last bracket vanishes. The Leibniz rule with respect to functions is obvious. For $\nabla^r$ and $\mathrm{Ad}$ the computation is identical, using the other two rules of \eqref{eq:transf}. On $S_{c,H}(M)$ the formula is the connection induced by the principal connection through $\mu_c$, since $d\mu_c(X)m=Xm$.
\end{proof}

\begin{theorem}[Compatibilities]\label{thm:compat}
Let $\nabla^s$ be as in Definition~\ref{def:spinconn}.
\begin{enumerate}[label=(\alph*)]
\item Under Proposition~\ref{prop:ClAd}, the $\mathrm{Ad}$-connection on $\Cl(\HM)$ is the extension of $\nabla^H$ as a derivation of the Clifford product (Lemma~\ref{lem:DHcompat}).
\item $\nabla^s_{Z}(\phi\Psi)=(\nabla^H_{Z}\phi)\Psi+\phi\,\nabla^s_{Z}\Psi$ for $\phi\in\sec\Cl(\HM)$ (resp.\ its complexification) and $\Psi$ a section of $\Cl^l_{\SpineH}(M)$, $S_H(M)$ or $S_{c,H}(M)$.
\item $\nabla^s_{Z}(\Psi a)=(\nabla^s_{Z}\Psi)a$ for constant $a\in\CRH$. Hence $\nabla^s$ preserves every $I_f(M)$ (real or complex), the decomposition of Theorem~\ref{thm:ideal}(iv), and commutes with the isomorphisms $R_{g^{-1}}$ of Theorem~\ref{thm:ideal}(iii).
\item \emph{(Compatibility of $\mu_c$.)} For every complex left $\CRH$-module $\mathcal M_{c,H}$, the connection of $S_{c,H}(M)=P_{\SpineH}(M)\times_{\mu_c}\mathcal M_{c,H}$ is given by \eqref{eq:nablas}. It satisfies (b), constant module maps $T:\mathcal M_{c,H}\to\mathcal M'_{c,H}$ induce parallel bundle maps, and the identification $\IcH\simeq S_{c,H}(M)$ of Theorem~\ref{thm:ideal}(ii) is connection-preserving. The same holds for $\mu_l$ and $S_H(M)$.
\item The fibre forms $\beta$ on $\IH$ and $h$ on $\IcH$ are parallel: $Z\beta(\Psi,\Psi')=\beta(\nabla^s_{Z}\Psi,\Psi')+\beta(\Psi,\nabla^s_{Z}\Psi')$ and similarly for $h$. Complex conjugation on $\IcH$ is parallel, so $(\IcH,\nabla^s)=\mathbb C\otimes(\IH,\nabla^s)$.
\item The volume element $\varpi$ is $\nabla^H$-parallel, and the chiral subbundles $I^\pm_H(M)$ are $\nabla^s$-parallel.
\item The pairing $I^\star_f(M)\otimes I_f(M)\to M\times\mathbb R$ is parallel: $Z(\Phi\Psi)=(\nabla^r_{Z}\Phi)\Psi+\Phi\,\nabla^s_{Z}\Psi$.
\item The curvature is $R^s(Z_1,Z_2)\Psi=\tfrac12\,\underline{\Omega}(Z_1,Z_2)\Psi$, where $\underline\Omega(Z_1,Z_2)\in\sec\Ext^2\HM$ is the bivector associated, by \eqref{eq:norm} (times $2$), with the curvature $R^H(Z_1,Z_2)$ of $\nabla^H$.
\end{enumerate}
\end{theorem}
\begin{proof}
(a) On a local section $\underline E_A=[\zeta,E_A]_{\mathrm{Ad}}$ the $\mathrm{Ad}$-connection gives $[\zeta,\tfrac12[\omega^\zeta_{Z},E_A]]=[\zeta,\Omega^\zeta(Z)E_A]=\nabla^H_{Z}\underline E_A$. Both connections are derivations of the Clifford product ($Z$ and $\tfrac12[\omega,\cdot]$ are derivations) agreeing on $\HM$ and on scalars, and $\Cl(\HM)$ is generated by $\HM$.
(b) Locally, $Z(\phi\psi)+\tfrac12\omega\phi\psi=(Z\phi+\tfrac12[\omega,\phi])\psi+\phi(Z\psi+\tfrac12\omega\psi)$.
(c) The connection acts on the left, and right multiplication by constants commutes with it. If $\Psi f=\Psi$ then $(\nabla^s\Psi)f=\nabla^s(\Psi f)=\nabla^s\Psi$.
(d) Theorem~\ref{thm:welldef}(b), with (b) and (c) above. A constant module map commutes with left multiplication by $\omega^\zeta$ and with $Z$.
(e) Locally $\tilde\psi\varphi=\beta(\psi,\varphi)\theta^*$ with $\theta^*$ constant; $Z(\tilde\psi\varphi)=\widetilde{Z\psi}\,\varphi+\tilde\psi\,Z\varphi$ and $\widetilde{\tfrac12\omega\psi}\varphi+\tilde\psi\tfrac12\omega\varphi=\tfrac12\tilde\psi(\tilde\omega+\omega)\varphi=0$ by Lemma~\ref{lem:lie}. For $h$ use also $\omega^\natural=\omega$ ($\omega$ is real). The same fact gives the parallelism of $\natural$.
(f) $\varpi$ is central in the even subalgebra, so $[\zeta,\varpi]_{\mathrm{Ad}}$ is independent of $\zeta$ by \eqref{eq:transf} and $\tfrac12[\omega,\varpi]=0$ (cf.\ Proposition~\ref{prop:parallel}(c)). $I^\pm_H(M)$ are the eigenbundles of the parallel endomorphism $\varpi\cdot$ (up to the sign $(-1)^{n(n-1)/2}$), by (b).
(g) $Z(\varphi\psi)=(Z\varphi-\tfrac12\varphi\omega)\psi+\varphi(Z\psi+\tfrac12\omega\psi)$.
(h) The local curvature of the principal connection $\tfrac12\omega$ is $F=d(\tfrac12\omega)+\tfrac12\omega\wedge\tfrac12\omega$ and $d\rho(F)$ is the curvature matrix of $\nabla^H$. By \eqref{eq:norm}, $F=\tfrac12\underline\Omega$, and it acts through $d\mu$, i.e.\ by left multiplication.
\end{proof}

\begin{proposition}[Explicit form for isotropic extensions]\label{prop:explicit}
Let $\nabla$ be any linear connection on $M$, $\nabla^H$ its isotropic extension, $e$ a positively oriented\footnote{Positive orientation is needed: if $\det A<0$, then $j(A)\in SO_H\setminus\SOeH$, so the Witt-induced frame of a negatively oriented frame does not belong to $P_{\SOeH}(M)=F^+(M)\times_j\SOeH$ (see the proof of Theorem~\ref{thm:canonical}).} local frame of $TM$ (not necessarily orthonormal) with $\nabla_{Z}e_l=\omega_l{}^k(Z)e_k$, and $\zeta$ a local section of $P_{\SpineH}(M)$ with $\Upsilon\circ\zeta=\iota(e)$ (it exists on contractible $U$. For $\Pcan$ one can take $\zeta=[e,1]$). Then
\begin{equation}\label{eq:omega-witt}
\omega^\zeta_{Z}=\sum_{k,l}\omega_l{}^k(Z)\,e_k\wedge\theta^l,\qquad \nabla^s_{Z}[\zeta,\psi]=\big[\zeta,\,Z(\psi)+\tfrac12\,\omega_l{}^k(Z)\,e_k\wedge\theta^l\,\psi\big]
\end{equation}
and on $\IH$, $\nabla^s_{Z}[\zeta,\alpha\fH]=[\zeta,(\nabla_{Z}\alpha+\tfrac12\operatorname{tr}\omega(Z)\,\alpha)\fH]$, where $\nabla_{Z}\alpha$ is the connection induced on forms, written in the frame $e$. Consequently, for $P=\Pcan$ and under the isomorphism $\Phi$ of Theorem~\ref{thm:fock}, $\nabla^s=\nabla\otimes1+1\otimes\nabla^{(1/2)}$, where $\nabla^{(1/2)}$ is the connection induced by $\nabla$ on $(\Ext^nTM)^{1/2}$. If $\nabla$ is compatible with a metric $\gring$ and $e$ is $\gring$-orthonormal, then $\operatorname{tr}\omega=0$, and under $\IH\simeq\Ext T^*M$ the spin connection is exactly the connection $\nabla$ induced on forms (for $\nabla=\Dring$: the Levi-Civita connection on forms). The curvature bivector is $\underline\Omega(Z_1,Z_2)=R^k{}_l(Z_1,Z_2)\,e_k\wedge\theta^l$.
\end{proposition}
\begin{proof}
In the Witt-induced frame the connection matrix of $\nabla^H$ is $j_*(\omega(Z))$, so \eqref{eq:omega-witt} follows from Lemma~\ref{lem:lie}. By the proof of Proposition~\ref{prop:lift}, $\tfrac12\omega_{Z}\,(\alpha\fH)=[\tfrac12\omega_{Z},\alpha]\fH+\alpha\,\tfrac12\omega_{Z}\fH=\big(-\omega_l{}^k(Z)\,\theta^l\wedge i_{e_k}\alpha+\tfrac12\operatorname{tr}\omega(Z)\alpha\big)\fH$, and $Z(\alpha_e)-\omega_l{}^k(Z)\theta^l\wedge i_{e_k}\alpha_e$ is the frame expression of $\nabla_{Z}\alpha$. Since $\nabla_{Z}e_*=\operatorname{tr}\omega(Z)e_*$ (Proposition~\ref{prop:parallel}(a)), the term $\tfrac12\operatorname{tr}\omega$ is the connection form of $\nabla^{(1/2)}$ in the frame $(e_*)^{1/2}$. For a $\gring$-orthonormal frame and a metric connection, $\operatorname{tr}\omega=0$ by Proposition~\ref{prop:parallel}(d).
\end{proof}

\section{Conclusions}\label{sec:conclusions}

We have developed the theory of spinor fields associated with the hyperbolic Clifford algebra $\RH=\Cl(\HV)$. The substitution $\mathbb R_{1,3}\to\RH$, $\Spine\to\SpineH$ in the classical construction of Mosna and Rodrigues \cite{MosnaRodrigues2004} leads to the left and right hyperbolic spin-Clifford bundles
\begin{align*}
\Cl^{l}_{\SpineH}(M) &= P_{\SpineH}(M)\times_l \RH, & l(u)x&=ux,\\
\Cl^{r}_{\SpineH}(M) &= P_{\SpineH}(M)\times_r \RH, & r(u)x&=xu^{-1},
\end{align*}
resting on the principal bundle $P_{\SpineH}(M)$ of \cite{NotteCuello2026} and on the hyperbolic groups $\Gamma_{H_{p,q}}$, $\mathrm{Pin}_{H_{p,q}}$, $\mathrm{Spin}_{H_{p,q}}$. The main differences with the Lorentzian case, all of them consequences of the neutrality of $\HV$, are the following:
\begin{enumerate}[label=(\roman*)]
\item $\RH\simeq M_{2^n}(\mathbb R)$ is split and the canonical primitive idempotent $\fH=\prod_k\tfrac12(1+\sigma_k\sigma_{n+k})$ is real and remains primitive in $\CRH\simeq M_{2^n}(\mathbb C)$. It generates the ideal $\RH\fH=\RH\theta^*\simeq\Ext V^*$ of Witten superfields, whose commutant is $\mathbb R$, so that the real spinor module carries no complex structure (Theorem~\ref{thm:fH}, Propositions~\ref{prop:idempotents}--\ref{prop:commutant});
\item every oriented manifold carries a canonical hyperbolic spin structure (Theorem~\ref{thm:canonical}), which gives an explicit realization of the absence of topological obstructions beyond orientability, established in \cite[\S1.1, pp.~173--175]{Karoubi1968} (see also \cite[\S\S2.3 and 2.8, pp.~10--17]{Gualtieri2004});
\item the spinor bundles are module bundles over $\Cl(\HM)$ and the canonical ideal spinor bundle is $\IH\simeq\Ext T^*M\otimes(\Ext^nTM)^{1/2}$, which reduces to $\Ext T^*M\simeq\Cl(M,\gring)$ in the presence of a metric of any signature (Proposition~\ref{prop:module}, Theorems~\ref{thm:ideal} and~\ref{thm:fock});
\item the isotropic extension of \emph{any} linear connection on $M$ is compatible with the hyperbolic metric (Lemma~\ref{lem:DHcompat}), the non-metricity being encoded in the non-parallelism of the splitting $\HM=H_+\oplus H_-$ (Remark~\ref{rem:relocation}). It induces a hyperbolic spin connection compatible with all the structures above (Theorems~\ref{thm:welldef} and~\ref{thm:compat}), which on $\IH$ is the connection induced on forms twisted by $(\Ext^nTM)^{1/2}$ (Proposition~\ref{prop:explicit}).
\end{enumerate}
These results provide the complete geometric framework (algebra, bundles, module structures, ideal spinor fields, and connections) upon which Dirac operators in the hyperbolic Clifford algebra are constructed, a topic that will be developed in a forthcoming article, together with the extension to non-orientable manifolds through the full group $\mathrm{Spin}_H$ (Remark~\ref{rem:orient}).

\section*{Declaration of generative AI and AI-assisted technologies in the manuscript preparation process}

During the preparation of this work, the authors used Claude (Anthropic) for improving the language and readability of the manuscript, checking the consistency of the bibliographic references against the original sources (in particular the citations of Karoubi, Gualtieri and Lawson--Michelsohn), suggesting revised wording for some proofs, and assisting with the \LaTeX{} editing. After using this tool/service, the authors reviewed and edited the content as needed and took full responsibility for the content of the article.

\section*{Acknowledgments}

JMHS thanks CNPq (Grant No. 307641/2022-8) for financial support.

\end{document}